\documentclass[12pt]{article}

\usepackage{amsmath, amssymb, amsthm}
\usepackage[margin=1in]{geometry}
\usepackage{graphicx}
\graphicspath{{outputs/danishmulti_allocation_results/}{danishmulti_allocation_results/}}
\usepackage{float}
\usepackage{xcolor}
\usepackage[round]{natbib}

\usepackage{booktabs,array}
\newcolumntype{P}[1]{>{\raggedright\arraybackslash}p{#1}}

\usepackage[colorlinks=true, linkcolor=blue, citecolor=blue, urlcolor=blue]{hyperref}

\newcommand{\licx}{\preceq_{\mathrm{icx}}}

\newcommand{\ind}{\mathbf{1}} 

\newcommand{\R}{\mathbb{R}}
\newcommand{\Rn}{\mathbb{R}^+}

\newcommand{\E}{\mathbb{E}}
\newcommand{\PP}{\mathbb{P}}
\newcommand{\dd}{\mathop{}\!\mathrm{d}} 

\newcommand{\ep}{\varepsilon}

\newcommand{\Jmarg}{\mathcal J_{\mathrm{marg}}}
\newcommand{\Jworst}{\mathcal J_{\mathrm{worst}}}

\newenvironment{practicalrevision}{}{}
\newenvironment{liquidityrevision}{}{}
\newenvironment{bluerevision}{}{}

\theoremstyle{plain} 
\newtheorem{Lemma}{Lemma}[section]

\newtheorem{Proposition}[Lemma]{Proposition}

\newtheorem{Theorem}[Lemma]{Theorem}

\theoremstyle{definition} 

\newtheorem{Example}[Lemma]{Example}
\newtheorem{Remark}[Lemma]{Remark}

\renewenvironment{proof}[1][Proof]{\noindent\textbf{#1.} }{\hfill$\square$\par\vspace{2mm}}

\numberwithin{equation}{section}
\allowdisplaybreaks 

\title{\Large \textbf{On the Expected Maximum Deficit \\ \vspace{1mm} and the Optimal Allocation of Reserves}}

\author{
\textbf{Claude Lef\`evre}\thanks{Universit\'e Libre de Bruxelles, D\'epartement de Math\'ematique, Campus de la Plaine C.P. 210, B-1050 Bruxelles, Belgium. Email: \texttt{claude.lefevre@ulb.be}}
\and 
\textbf{Pierre Zuyderhoff}\thanks{University of Nottingham, Department of Mathematical Sciences, 199 Taikang East Road, Yinzhou District, Ningbo, China. Email: \texttt{z2022170@nottingham.edu.cn}}
}
\date{}

\begin{document}

\maketitle

\vspace{2mm}
\begin{abstract}
\begin{liquidityrevision}
\noindent Let $L=(L_s)_{0\le s\le t}$ be a cumulative net-loss process and let $M_t=\sup_{0\le s\le t}L_s$. For a candidate reserve $u$ and a distortion function $g$, define
$\mathcal D_g^{(t)}(u)=\int_u^\infty g(\PP(M_t>v))\dd v$.
\begin{bluerevision}
This function measures the tail-weighted residual severity of the largest cumulative loss over the horizon. We derive three monetary risk measures: its value at zero reserve and two measures based on fixed and proportional deficit tolerances. Under a concave distortion, the zero-reserve and proportional-tolerance measures are coherent, while the fixed-tolerance measure is convex. Given a reserve budget distributed among several units, we analyze two allocation criteria: the sum of the units' distorted deficits, which depends only on their marginal distributions, and a distorted expectation of the largest residual unit deficit, which uses their joint distribution. We characterize the optimal allocations and corresponding near-optimal sets. A two-unit exponential benchmark shows that the optimum can be unique and unequal. We then apply both criteria to the Building, Contents and loss-of-Profits components of the public Danish large-fire data. Under the baseline predictive model, both yield interior allocations and localized 0.5\%-optimal regions, while sensitivity analyses illustrate model uncertainty. We also provide Monte Carlo estimators, conditional counterparts and a review-date allocation formulation.
\end{bluerevision}
\end{liquidityrevision}
\end{abstract}

\vspace{4mm}


\section{Introduction}
\label{sec:intro}

\begin{liquidityrevision}
An insurer's cumulative net loss evolves throughout a reporting period, and its most adverse value need not occur at the reporting horizon. Reserve management may therefore depend on the largest loss reached over the period rather than on the terminal loss alone. This paper develops monetary risk measures based on the expected maximum deficit (EMD) and studies the allocation of a given reserve budget among allocation units, such as business lines or coverage groups.

Fix a horizon $t$. At time $s\le t$, let $L_s$ denote cumulative net loss since the beginning of the horizon. Depending on the application, $L_s$ may represent accumulated underwriting net loss, the decline in economic net assets from their initial value, or cumulative dated cash outflows less cash inflows. Define
\[
M_t=\sup_{0\le s\le t}L_s.
\]
For a candidate reserve $u$, the residual maximum deficit is $(M_t-u)_+$, where $x_+=\max\{x,0\}$. It is zero when the reserve covers the largest cumulative net loss and otherwise records the remaining shortfall. Because $t$ is included in the supremum, the terminal loss is part of the criterion, but the maximum may be attained earlier. EMD records the size of the maximum, not when or why it occurs.

We study the expected deficit and distorted versions. A \emph{distortion function} is a nondecreasing map $g:[0,1]\to[0,1]$ satisfying $g(0)=0$ and $g(1)=1$; it is applied here to exceedance probabilities. The distorted EMD is
\[
\mathcal D_g^{(t)}(u)=\int_u^\infty g\!\left(\PP(M_t>v)\right)\dd v.
\]
Concavity implies $g(p)\ge p$, so the distorted EMD gives at least as much weight to tail probabilities as the ordinary EMD, obtained from the identity distortion.

\begin{bluerevision}
The distorted EMD should first be viewed as the reserve-indexed function $u\mapsto\mathcal D_g^{(t)}(u)$ rather than as a single risk-measure value. Three risk-measure values are obtained from this function. The zero-reserve value is its evaluation at $u=0$ and introduces no tolerance parameter. The fixed-tolerance value is the smallest capital adjustment that reduces the residual deficit below a prescribed monetary allowance. The proportional-tolerance value is defined similarly, but its allowance scales with the candidate capital. Even with the same distortion, two loss processes can have the same zero-reserve value but different reserve-indexed functions, so their tolerance-based capital adjustments can differ. The distortion governs the weighting of exceedance probabilities, whereas the monetary and proportional tolerances express risk appetite and are not identified from the loss data or the resources already held.
\end{bluerevision}

\begin{bluerevision}
Throughout Sections~\ref{sec:deficit}--\ref{sec:allocation}, the candidate reserve $u$ is a deterministic amount subtracted from the loss process. It is not automatically an accounting technical provision, regulatory capital requirement or liquidity holding. Its operational meaning depends on how $L$ is specified. We develop three interpretations. For an occurrence-based underwriting process, EMD measures the largest cumulative underwriting net loss and provides an internal reserve-management criterion. For an economic balance-sheet process, it measures the largest deterioration in net assets and may inform internal economic-capital management. For dated cash flows, it measures the largest cumulative payment--receipt gap and has a liquidity interpretation. In every case, a further management policy is needed to translate any of these three risk-measure values into resources to be held or raised.
\end{bluerevision}

Related actuarial literature places EMD among ruin-severity criteria. Classical ruin theory supplements ruin probability with measures of severity, including expected discounted penalty functions \citep{GerberEtAl1987,GerberShiu1998} and maximum-severity quantities \citep{Dickson1998,ThampiEtAl2007,LiLu2013}. \citet{TrufinEtAl2011} derive a monetary requirement from a bound on ruin probability. \citet{Loisel2005} develops differentiation formulas for risk-process functionals and their allocation, while \citet{LoiselTrufin2014} use the expected area in red to control the time-integrated negative surplus and allocate a company-wide risk limit. The area in red combines the depth and duration of every period below zero. EMD uses a different temporal aggregation: it measures the residual severity of the worst shortfall on a path. The expected area in red and EMD therefore describe complementary aspects of the surplus trajectory. Finite-horizon EMD was considered by \citet{Jiang2015}, and \citet{LefevreEtAl2017} established stochastic comparisons in the classical risk model.

The paper also builds on monetary risk measures for stochastic processes. \citet{CheriditoEtAl2005} provide a general framework for coherent and convex risk measures on paths. \citet{ShimizuTanaka2018} formulate Gerber--Shiu process risk measures in an asset--liability management and going-concern solvency setting, where a barrier breach has a balance-sheet interpretation. Here, distortion risk measurement is applied to the running maximum of a cumulative net-loss process. Related conditional foundations appear in \citet{FollmerPenner2006} and \citet{BieleckiEtAl2025}.

A dated-cash-flow specification connects EMD with liquidity-risk measures based on financing constraints or stressed cash resources \citep{AcerbiScandolo2008,SkoglundChen2012,ContEtAl2020}. In that specification, EMD measures the residual depth of the largest cumulative net cash outflow. Liquidity remains one interpretation of distorted EMD, not its defining application here.

\begin{bluerevision}
We next study how to split an exogenous internal reserve budget among allocation units. OR-based ruin allocation, in which at least one unit breaches its reserve, is studied by \citet{Loisel2004}; \citet{GongEtAl2012} consider both OR and AND criteria, the latter requiring every unit to breach, possibly at different times. Adaptive subset-ruin constraints are studied by \citet{DelsingEtAl2019}. As a criterion for failure of the pooled portfolio, OR ruin may be conservative when an isolated unit breach can be met by central resources or subsequent reallocation, whereas AND ruin may be permissive because pooled ruin can occur without every unit breaching. This motivates severity-sensitive criteria rather than reliance on either binary extreme. Method~1 minimizes the sum of the unit EMDs and uses only their marginal laws. Method~2 minimizes the distorted expected largest residual unit deficit; equivalently, it integrates distorted OR probabilities over all additional deficit levels and therefore uses the joint law. Neither Method~1 nor Method~2 replaces a separate assessment of pooled adequacy. They also differ from Euler or Aumann--Shapley contribution decompositions of an aggregate risk-measure value \citep{Tsanakas2009,DhaeneEtAl2012}: the total budget is given and the optimization determines its split. Identifying that budget with a prescribed solvency-capital amount, if desired, is an external decision.
\end{bluerevision}

\begin{bluerevision}
Against this background, the paper's main contributions are the fixed- and proportional-tolerance risk measures and the characterizations of the allocations selected by Methods~1 and~2. For a concave distortion, the measure obtained at zero reserve is coherent, as is the proportional-tolerance measure. The fixed-tolerance measure is convex. The paper also gives stochastic-comparison results under explicit ordering assumptions and discusses the interpretation of all three risk-measure values. Finally, a Monte Carlo plug-in estimator based on simulated predictive paths estimates the three risk-measure values and both allocation objectives without requiring a closed-form ruin probability, while leaving their dependence on the predictive model visible.
\end{bluerevision}

For the empirical illustration, we use the public Danish fire-loss data \texttt{danishmulti}, which contain dated losses divided into Building, Contents and loss-of-Profits components \citep{EmbrechtsEtAl1997,DutangCharpentier2026}. We use 1980--1988 for model specification and retain 1989--1990 for predictive diagnostics. The annual event count is represented by a Poisson--gamma mixture, while complete event records are resampled to preserve occurrence timing and common-event dependence. Because the data contain no premiums, the loss paths use a fitted expected-cost funding schedule for the observed large-loss layer. The fitted Method~1 and Method~2 objectives have interior allocations and localized near-optimal regions, while funding and dependence sensitivities prevent interpreting them as observed or prescribed reserve allocations. A short final section studies updating when new information arrives; dynamic control and recursive time consistency remain outside its scope.

\begin{bluerevision}
Section~\ref{sec:deficit} defines EMD and its distorted form. Section~\ref{sec:riskmeasures} derives the zero-reserve running-maximum, fixed-tolerance and proportional-tolerance risk measures, and Section~\ref{sec:allocation} studies Methods~1 and~2. Section~\ref{sec:implementation} gives the actuarial interpretations, Monte Carlo estimator and three-coverage fire-loss illustration. Section~\ref{sec:dynamic} studies conditional versions of the three risk measures and review-date reallocation. Compound-Poisson arrivals and exponential severities are used as analytical benchmarks; the empirical implementation instead combines a mixed-Poisson count with jointly resampled event records.
\end{bluerevision}
\end{liquidityrevision}


\section{Mathematical framework and the expected maximum deficit}
\label{sec:deficit}

We work on a filtered probability space $(\Omega,\mathcal F,(\mathcal F_s)_{s\ge0},\PP)$ satisfying the usual conditions and supporting regular conditional distributions. Processes used for infinite-horizon limits are defined on $[0,\infty)$; for a fixed finite horizon $t$, write $\mathcal H^0_t$ for their restrictions to $[0,t]$. We assume $\mathcal F_0$ is trivial up to null sets, so the initial state of an adapted process is deterministic. Given $a,b\in\R$, $aX+b$ denotes the process $(aX_s+b)_{0\le s\le t}$. To retain the notation used below, set $\mathcal H^0:=\mathcal H^0_t$, i.e.,
\[
\mathcal{H}^0=\big\{X: [0,t]\times\Omega\to\R \mid X \text{ is c\`adl\`ag and } \mathcal{F}_s\text{-adapted,} \, 0\leq s\leq t\big\}.
\]
The subspace of bounded processes is defined as
\[
\mathcal{H}^\infty=\left\{X\in\mathcal{H}^0 \;\middle|\; \left\|\sup_{0\le s\le t}|X_s|\right\|_\infty<\infty\right\}.
\]

Consider a continuous-time insurance risk process $R \in \mathcal{H}^0$ within this space for which $R_s$ represents the surplus at time $s$, with an initial reserve $R_0 = u \in \R$. The overall net loss up to time $s$ is measured by the random variable
\begin{equation}\label{eq:netloss}
L_s=u-R_s, \quad s\geq 0,
\end{equation}
with $L_0=0$. In particular, for the classical compound Poisson risk model,
\begin{equation}\label{eq:poisson}
L_s=\sum_{i=1}^{N_s} Y_i - cs,
\end{equation}
where $\{N_s, \, s\in \Rn\}$ is a Poisson process counting claim arrivals up to time $s$, the $Y_i$ are successive i.i.d.\ claim amounts (distributed as $Y$), and $c$ is the constant premium rate.

By the definition of $L_s$ in \eqref{eq:netloss}, the maximum overall net loss up to time $t$ is
\begin{equation}\label{eq:maxloss}
M_t=\sup\{L_s, \, 0\le s\le t\},
\end{equation}
and the ruin probability over $(0,t)$, denoted $\psi_t(u)$, is simply
\begin{equation}\label{eq:ruinprob}
\psi_t(u)=\PP(M_t>u) \quad (\text{which is also denoted } {\bar F}_{M_t}(u)).
\end{equation}

Our primary focus is the \emph{expected maximum deficit up to time $t$}, which combines the probability and depth of an exceedance. It is defined by
\begin{equation}\label{eq:deficit}
\mathcal{D}^{(t)}(u) := \E\big[(M_t-u)_+\big]=\int_u^\infty (m-u)\dd\PP(M_t\le m).
\end{equation}
After a simple integration by parts, \eqref{eq:deficit} can be rewritten as
\begin{equation}\label{eq:stoploss}
\mathcal{D}^{(t)}(u)=\int_u^\infty \psi_t(v)\dd v,
\end{equation}

As a function of $u$, this is the stop-loss transform of the running maximum $M_t$, expressed through its ruin-probability tail. Its connection with stop-loss and increasing-convex order is classical \citep{Muller1996,MullerStoyan2002}; infinite-horizon transforms of ruin probabilities are studied, for example, by \citet{ChengPai2003} and \citet{Tsai2006}.

The proposition below ensures that $\mathcal{D}^{(t)}(u)$ is consistent with the increasing convex ordering $\licx$ on claim amounts, meaning it penalizes risks that are larger and more variable.

\begin{Proposition}\label{prop:lefevreetal}
Let $R_s^{(j)}$, $j=1,2$, be two risk processes as in \eqref{eq:poisson}, each characterized by a Poisson intensity $\lambda_j$, claim amounts distributed as $Y^{(j)}$, and premium rate $c_j$. If $\lambda_1\le \lambda_2$ and $c_1\ge c_2$, then $Y^{(1)} \licx Y^{(2)}$ implies
\begin{equation}\label{eq:deficit_ineq}
\mathcal{D}_1^{(t)}(u)\le \mathcal{D}_2^{(t)}(u), \quad t, u\geq 0.
\end{equation}
\end{Proposition}
\begin{proof}
Under the stated assumptions, \citet{LefevreEtAl2017} established that $M_t^{(1)} \licx M_t^{(2)}$ for $t\geq 0$. Therefore, inequality \eqref{eq:deficit_ineq} follows from \eqref{eq:stoploss} with \eqref{eq:ruinprob} and the property that the standard stop-loss transform preserves the order $\licx$.
\end{proof}

To account for tail risk aversion, we generalize the expected maximum deficit by evaluating it under a non-decreasing distortion function $g: [0,1] \to [0,1]$ with $g(0)=0$ and $g(1)=1$. This tail transformation follows the distorted-premium and distortion-risk-measure literature \citep{Wang1996,WangYoungPanjer1997,WirchHardy2002}. For a general random variable $Z$, the distorted expectation is defined via the Choquet integral as
\begin{equation}\label{eq:choquet}
\E_g[Z] = \int_0^\infty g\Big(\PP(Z > x)\Big)\dd x + \int_{-\infty}^0 \Big[g\Big(\PP(Z > x)\Big) - 1\Big]\dd x.
\end{equation}
We define the distorted expected maximum deficit up to time $t$ as $\mathcal{D}_g^{(t)}(u) := \E_g\big[(M_t-u)_+\big]$. Because the positive part $(M_t-u)_+$ is a non-negative random variable, $\PP\big((M_t-u)_+ > x\big) = 1$ for all $x < 0$. Since $g(1) = 1$, the integrand in the negative domain evaluates to zero, and the Choquet integral reduces to:
\begin{equation}\label{eq:distorted_deficit}
\mathcal{D}_g^{(t)}(u) = \int_0^\infty g\Big(\PP\big((M_t-u)_+ > x\big)\Big)\dd x.
\end{equation}
Applying the change of variable $v = u + x$ and recalling \eqref{eq:ruinprob}, this simplifies to the distorted stop-loss transform:
\begin{equation}\label{eq:distorted_stoploss}
\mathcal{D}_g^{(t)}(u) = \int_u^\infty g\big(\psi_t(v)\big)\dd v.
\end{equation}
When $g(x)=x$, $\mathcal{D}_g^{(t)}(u)$ reduces to the standard expected maximum deficit $\mathcal{D}^{(t)}(u)$. 

\begin{liquidityrevision}
\begin{Remark}[Comparison with terminal loss]\label{rem:terminal_intrahorizon}
Because $M_t\ge L_t$ pathwise, monotonicity of the positive part and of the distorted expectation gives
\begin{equation}\label{eq:terminal_deficit_dominance}
\mathcal D_g^{(t)}(u)
=\E_g[(M_t-u)_+]
\ge \E_g[(L_t-u)_+]
=\int_u^\infty g\!\left(\PP(L_t>v)\right)\dd v.
\end{equation}
The event identity
\begin{equation}\label{eq:terminal_intrahorizon_split}
\{M_t>u\}=\{L_t>u\}\,\dot\cup\,\{M_t>u, L_t\le u\}
\end{equation}
can help interpret an exceedance. It is not a decomposition of the distorted deficit: $\mathcal D_g^{(t)}$ depends on the loss process only through the distribution of $M_t$. Consequently, two processes with the same running-maximum distribution receive the same EMD even if their terminal behavior differs. Distinguishing whether a maximum persists to the horizon requires additional information on $L_t$ or on the joint law of $(M_t,L_t)$.
\end{Remark}
\end{liquidityrevision}

\begin{Lemma}[Basic properties of the distorted deficit]\label{lem:deficit_properties}
Assume $\mathcal D_g^{(t)}(0)<\infty$ and let $g$ be any nondecreasing distortion. Then $u\mapsto\mathcal D_g^{(t)}(u)$ is nonincreasing, convex, continuous and $1$-Lipschitz. It is differentiable almost everywhere, with
\begin{equation}\label{eq:deficit_derivative}
\frac{\dd}{\dd u}\mathcal D_g^{(t)}(u)=-g\bigl(\psi_t(u)\bigr).
\end{equation}
Moreover, since $L_0=0$ and hence $M_t\ge0$,
\begin{equation}\label{eq:negative_linear}
\mathcal D_g^{(t)}(u)=\mathcal D_g^{(t)}(0)-u,\qquad u\le0.
\end{equation}
For the identity distortion and every $u$ such that $\psi_t(u)>0$,
\begin{equation}\label{eq:prob_severity}
\mathcal D^{(t)}(u)=\psi_t(u)\,\E[M_t-u\mid M_t>u].
\end{equation}
Strict decrease is not a general property; it holds on intervals on which $\psi_t(u)>0$ and $g(p)>0$ for $p>0$.
\end{Lemma}

\begin{proof}
The integrand $g\circ\psi_t$ is nonincreasing and takes values in $[0,1]$. Hence the integral in \eqref{eq:distorted_stoploss} is nonincreasing and has almost-everywhere derivative \eqref{eq:deficit_derivative}; that derivative is nondecreasing, which proves convexity. For $u_1<u_2$,
\[
0\le \mathcal D_g^{(t)}(u_1)-\mathcal D_g^{(t)}(u_2)
=\int_{u_1}^{u_2}g(\psi_t(v))\dd v\le u_2-u_1,
\]
proving continuity and the Lipschitz bound. Since $\psi_t(v)=1$ for $v<0$, splitting the integral at zero gives \eqref{eq:negative_linear}. Finally, \eqref{eq:prob_severity} follows by conditioning $(M_t-u)_+$ on $\{M_t>u\}$.
\end{proof}

\begin{practicalrevision}
When the net profit and integrability conditions ensure that $M_\infty:=\sup_{s\ge0}L_s$ has finite distorted expectation, define the infinite-horizon deficit by
\begin{equation}\label{eq:infinite_horizon_deficit}
\mathcal D_g^{(\infty)}(u)
:=\E_g[(M_\infty-u)_+],
\qquad u\in\R.
\end{equation}
Here $\E_g$ is the distorted expectation defined in \eqref{eq:choquet}; the calligraphic notation $\mathcal E_{g,s}$ is reserved for the conditional operator introduced in Section~\ref{subsec:timeconsistency}. If $g$ is continuous, the tail-integral representation and monotone convergence give
\[
\mathcal D_g^{(\infty)}(u)
=\lim_{t\to\infty}\mathcal D_g^{(t)}(u).
\]
\end{practicalrevision}

Since a Choquet integral generated by a concave distortion function preserves the increasing convex order, the consistency result extends naturally.

\begin{Proposition}\label{prop:distorted_lefevreetal}
Let $L^{(1)}$ and $L^{(2)}$ be two aggregate net loss processes in $\mathcal{H}^0$ such that their maximum overall net losses satisfy $M_t^{(1)} \licx M_t^{(2)}$ for a given $t \ge 0$. If the distortion function $g$ is concave, then their distorted expected maximum deficits verify
\begin{equation}\label{eq:dist_icx}
\mathcal{D}_{g,1}^{(t)}(u) \le \mathcal{D}_{g,2}^{(t)}(u), \quad \text{for all } u \in \R.
\end{equation}
\end{Proposition}

\begin{proof}
Since the function $f(x) = (x-u)_+$ is non-decreasing and convex for any $u \in \R$, it follows from the properties of the increasing convex order that $M_t^{(1)} \licx M_t^{(2)}$ implies $(M_t^{(1)}-u)_+ \licx (M_t^{(2)}-u)_+$. It is a standard result \citep[see, e.g.,][]{WirchHardy2002} that the Choquet integral $\E_g[\cdot]$ with a concave distortion function $g$ is monotone with respect to the order $\licx$. Applying the operator $\E_g[\cdot]$ to both sides yields the inequality \eqref{eq:dist_icx}.
\end{proof}


\section{Static risk measures}
\label{sec:riskmeasures}

\begin{liquidityrevision}
This section derives three static monetary risk measures from the expected maximum deficit over $[0,t]$. Here and throughout Sections~\ref{sec:riskmeasures}--\ref{sec:allocation}, ``capital'' has its mathematical risk-measure meaning: a deterministic cash amount added to a position, equivalently subtracted from a loss process. Operational interpretation is deferred to Section~\ref{sec:implementation}.
\end{liquidityrevision}

Before detailing specific functionals, we first recall the formal definitions of convex and coherent risk measures on the space of unbounded c\`{a}dl\`{a}g processes $\mathcal{H}^0$, as established by \citet{CheriditoEtAl2005}. A functional $\rho: \mathcal{H}^0 \to \R \cup \{+\infty\}$ is called a \emph{convex risk measure} if it satisfies the following conditions for any $X, Y \in \mathcal{H}^0$:
\begin{enumerate}
\item[(i)] \textbf{Finiteness:} $\rho(X) \in \R$ for all $X \in \mathcal{H}^\infty$.
\item[(ii)] \textbf{Translation Invariance:} $\rho(X+c) = \rho(X) + c$ for any constant $c \in \R$.
\item[(iii)] \textbf{Monotonicity:} If $X_s \le Y_s$ almost surely for all $s \in [0,t]$, then $\rho(X) \le \rho(Y)$.
\item[(iv)] \textbf{Convexity:} $\rho(\alpha X + (1-\alpha) Y) \le \alpha \rho(X) + (1-\alpha) \rho(Y)$ for any $\alpha \in [0,1]$.
\end{enumerate}
Furthermore, the functional $\rho$ is a \emph{coherent risk measure} if it is a convex risk measure that additionally satisfies:
\begin{enumerate}
\item[(v)] \textbf{Positive Homogeneity:} $\rho(\lambda X) = \lambda \rho(X)$ for any $\lambda \ge 0$.
\end{enumerate}
Under positive homogeneity, convexity is mathematically equivalent to subadditivity, i.e., $\rho(X+Y) \le \rho(X) + \rho(Y)$.

We focus on the risk measures defined from the expected maximum deficit. Multiplying a monetary risk measure by a scalar generally changes its cash-additivity coefficient; any normalization must therefore be specified as part of the acceptance criterion rather than treated as an innocuous rescaling.

\begin{Remark}[Initial-state anchored scaling]\label{rem:anchored_scaling}
The preceding warning admits a useful qualification on a process space. Since $\mathcal F_0$ is trivial, the initial-value functional
\[
\rho_0(X)=X_0
\]
is real-valued and coherent. Let $\rho$ be any convex risk measure on $\mathcal H^0$. For $\gamma\in(0,1]$, define
\begin{equation}\label{eq:anchored_scaling}
\rho^{\langle\gamma\rangle}(X)
=(1-\gamma)\rho_0(X)+\gamma\rho(X),
\end{equation}
and set $\rho^{\langle0\rangle}=\rho_0$. The class of convex risk measures is closed under convex combinations, so $\rho^{\langle\gamma\rangle}$ is convex; if $\rho$ is coherent, it is coherent as well. This can also be checked axiom by axiom, since both components have the same cash-additivity coefficient and the weights are nonnegative and sum to one.

For the normalized net-loss process $L_0=0$, equation~\eqref{eq:anchored_scaling} reduces to
\begin{equation}\label{eq:anchored_scaling_loss}
\rho^{\langle\gamma\rangle}(L)=\gamma\rho(L).
\end{equation}
Thus a downward scalar rescaling on the class of zero-initial-state loss processes has a valid cash-additive extension to the full process space. It is the anchored extension in \eqref{eq:anchored_scaling}, rather than the raw functional $\gamma\rho$, that remains monetary on general processes.

The restriction $\gamma\le1$ cannot be removed as a general closure result. For example, let $\rho(X)=X_t$, choose $0<\varepsilon<t$, and consider the deterministic paths $X_s=0$ and $Y_s=(1-s/\varepsilon)_+$. Then $X_s\le Y_s$ for every $s$, while $X_0=0$, $Y_0=1$ and $X_t=Y_t=0$. Extending \eqref{eq:anchored_scaling} to $\gamma>1$ gives
\[
\rho^{\langle\gamma\rangle}(Y)=1-\gamma<0
=\rho^{\langle\gamma\rangle}(X),
\]
which violates monotonicity. Exceptional risk measures may permit larger coefficients, but coherence or convexity alone does not guarantee this.
\begin{bluerevision}
The observation applies to $\rho_g^{(t)}$, $\tilde\rho_{A,g}^{(t)}$ and $\rho_{\delta,g}^{(t)}$. For $\tilde\rho_{A,g}^{(t)}$ and $\rho_{\delta,g}^{(t)}$, external anchored scaling is not the same operation as recalibrating $A$ or $\delta$, which changes the underlying acceptance criterion.
\end{bluerevision}
\end{Remark}

\subsection{A coherent measure of the expected deficit}
\label{subsec:coherent}

\begin{bluerevision}
Consider any distortion function $g: [0,1] \to [0,1]$, which is non-decreasing with $g(0)=0$ and $g(1)=1$. For a generic continuous-time process $X \in \mathcal{H}^0$, define the static distortion risk measure over the horizon $t$ as the Choquet integral of its running supremum. Extending the representation from Section~\ref{sec:deficit} to the entire real line gives:
\end{bluerevision}
\begin{equation}\label{eq:risk_measure}
\rho_g^{(t)}(X) = \E_g\Big[\sup_{0 \le s \le t} X_s\Big].
\end{equation}

Related path-dependent process risk measures appear explicitly in \citet[Section~5.1]{CheriditoEtAl2005}, who study VaR and AVaR of the running minimum of a c\`adl\`ag surplus process and apply them to the Cram\'er--Lundberg model. Under the loss convention $L=-R$, their AVaR measure is the TVaR-distortion case of \eqref{eq:risk_measure}; the present formulation embeds it in the wider class of concave distortions and supplies the running-maximum functional used in the fixed-tolerance measure $\tilde\rho_{A,g}^{(t)}$ and the proportional-tolerance measure $\rho_{\delta,g}^{(t)}$.

\begin{Proposition}\label{prop:coherentriskmeasure}
Assume that the distortion function $g$ is concave. The static risk measure $\rho_g^{(t)}$ defined in \eqref{eq:risk_measure} is a coherent risk measure on $\mathcal{H}^0$.
\end{Proposition}
\begin{proof}
First, we verify the finiteness condition. For any $X \in \mathcal{H}^\infty$, its running supremum $\sup_{0 \le s \le t} X_s$ is a bounded random variable. Because the distortion function $g$ maps to $[0,1]$, the Choquet integral of a bounded random variable evaluates to a finite real number, ensuring $\rho_g^{(t)}(X) \in \R$.

The functional $\rho_g^{(t)}(\cdot)$ operates as the Choquet integral $\E_g[\cdot]$ evaluated on the running supremum $\sup_{0 \le s\le t}(\cdot)$. The supremum is monotone, translation invariant and positively homogeneous pointwise. For a concave distortion function $g$, the Choquet integral is well known to be a coherent risk measure on the space of random variables \citep[see, e.g.,][]{WirchHardy2002}. For completeness, subadditivity of the composition follows from
\[
\sup_{0\le s\le t}(X_s+Y_s)
\le \sup_{0\le s\le t}X_s+\sup_{0\le s\le t}Y_s
\]
and then from monotonicity and subadditivity of $\E_g$. Translation invariance and positive homogeneity follow from the corresponding pointwise identities for the running supremum, and convexity follows from subadditivity and positive homogeneity. Thus each coherence axiom is verified directly; no general composition principle is being invoked.
\end{proof}

\begin{bluerevision}
For the aggregate net-loss process $L$, the maximum $M_t=\sup_{0\le s\le t}L_s$ is nonnegative, so the negative part of the Choquet integral vanishes. Using $\psi_t(u)=\PP(M_t>u)$ and \eqref{eq:distorted_stoploss}, the value of \eqref{eq:risk_measure} on $L$ is the distorted expected maximum deficit at zero candidate reserve:
\end{bluerevision}
\begin{equation}\label{eq:risk_zero}
\rho_g^{(t)}(L) = \int_0^\infty g\big(\psi_t(v)\big)\dd v := \mathcal{D}_g^{(t)}(0), \quad t\geq 0.
\end{equation}
This identity formally links the actuarial concept of ruin severity to the axiomatic framework of coherent risk measures.

\begin{practicalrevision}
For the risk measure in \eqref{eq:risk_measure}, if $\E_g[M_\infty]<\infty$ and the distortion $g$ is continuous, define its infinite-horizon value by
\begin{equation}\label{eq:infinite_horizon_definition}
\rho_g^{(\infty)}(L)
:=\E_g[M_\infty]
=\mathcal D_g^{(\infty)}(0)
=\lim_{t\to\infty}\rho_g^{(t)}(L).
\end{equation}
\end{practicalrevision}

\begin{bluerevision}
For the TVaR distortion $g(x)=\min(x/\alpha,1)$, \eqref{eq:risk_measure} is AVaR of the running maximum, the loss-convention counterpart of the process measure in \citet[Section~5.1]{CheriditoEtAl2005}. The indicator distortion instead gives the corresponding VaR and is not concave, so coherence is not guaranteed. These familiar cases locate \eqref{eq:risk_measure} within the existing process-risk literature; the fixed-tolerance measure $\tilde\rho_{A,g}^{(t)}$ and proportional-tolerance measure $\rho_{\delta,g}^{(t)}$ are developed below.
\end{bluerevision}

\begin{Example}\label{ex:coherent_exp}
To establish an analytical benchmark for this coherent functional, consider an infinite-horizon ($t=\infty$) compound Poisson risk process with exponentially distributed claim severities $Y \sim \mathrm{Exp}(1/\mu)$ and a premium rate $c$ sustaining a positive safety loading. The same exponential Cram\'er--Lundberg setting is used by \citet[Section~5.1]{CheriditoEtAl2005} to illustrate VaR and AVaR process risk measures. The ultimate ruin probability decays exponentially according to $\psi(u) = (1-\mu R)e^{-Ru}$ for $u\ge0$, where the adjustment coefficient is defined as $R = \frac{1}{\mu}(1 - \frac{\lambda\mu}{c})$.
\\
For a general distortion function $g$, the coherent risk measure evaluates to $\rho_g^{(\infty)}(L) = \mathcal{D}_g^{(\infty)}(0) = \int_0^\infty g\left((1-\mu R)e^{-Rv}\right) \dd v$.
To obtain closed-form solutions, we particularize the distortion to a proportional hazard distortion $g(x) = x^k$ with $k \in (0,1]$, ensuring the necessary concavity. The expected maximum deficit functional yields a closed-form solution:

Set $a=1-\mu R$ and $C=a^k/(kR)$. Lemma~\ref{lem:deficit_properties} gives the deficit function:
\begin{equation}\label{eq:exp_deficit_piecewise}
\mathcal D_g^{(\infty)}(u)=
\begin{cases}
C-u,&u<0,\\[1mm]
C e^{-kRu},&u\ge0.
\end{cases}
\end{equation}

\begin{bluerevision}
Evaluating the deficit function at $u=0$ yields the running-maximum risk-measure value $\rho_g^{(\infty)}(L)$:
\end{bluerevision}
\begin{equation}\label{eq:exp_coherent}
\rho_g^{(\infty)}(L) = \frac{(1-\mu R)^k}{kR}=C.
\end{equation}
To benchmark the structural behavior of this measure, we isolate the undistorted Expected Maximum Deficit (EMD) framework by setting $g(x)=x$ (i.e., $k=1$). The coherent measure reduces to $\rho^{(\infty)}(L) = \frac{1-\mu R}{R}$.
\end{Example}

\subsection{Fixed- and proportional-tolerance monetary risk measures}
\label{subsec:implicit}

To integrate tail risk aversion for a generic continuous-time loss process $X \in \mathcal{H}^0$, one might first consider the raw functional
\[
\E_g\!\left[\left(\sup_{0\le s\le t}X_s\right)_+\right].
\]
Although the positive part preserves convexity and monotonicity, this functional is not cash additive: injecting deterministic capital reduces the expected deficit nonlinearly. It is therefore a convex penalty, not a monetary risk measure in isolation.

\subsubsection{A convex measure with a fixed risk tolerance}
\label{subsubsec:convex}

Because the raw functional is not cash additive, we use it instead to define an acceptance set in which the distorted expected positive shortfall is bounded by a prescribed deterministic tolerance $A>0$:
\begin{equation}\label{eq:acceptance}
\mathcal{A}_{A,g}^{(t)} = \left\{ X \in \mathcal{H}^0 \;\middle|\; \E_g\Big[\Bigl(\sup_{0\le s\le t} X_s\Bigr)_+\Big] \le A \right\},
\end{equation}
where the expectation $\E_g$ is evaluated via the Choquet integral with respect to a concave distortion function $g$.

The fixed parameter $A$ has monetary units and bounds the distorted expected depth of the maximum shortfall. Its interpretation and calibration are deferred to Section~\ref{sec:implementation}.

Following the axiomatic structure of convex risk measures, the associated risk measure is then defined as the minimal capital $u \in \R$ that must be injected into the portfolio so that the shifted loss process $(X - u)$ belongs to this static acceptance set:
\begin{equation}\label{eq:mincap}
\tilde{\rho}_{A,g}^{(t)}(X) = \inf \left\{ u \in \R : (X - u) \in \mathcal{A}_{A,g}^{(t)} \right\}.
\end{equation}

\begin{Proposition}[Capital below zero]\label{prop:negative_capital}
Suppose $\mathcal D_g^{(t)}(0)<\infty$. For the net-loss process $L$,
\begin{equation}\label{eq:universal_negative_capital}
\tilde\rho_{A,g}^{(t)}(L)=\mathcal D_g^{(t)}(0)-A
\qquad\text{whenever }A\ge\mathcal D_g^{(t)}(0).
\end{equation}
\end{Proposition}
\begin{proof}
By \eqref{eq:negative_linear}, $\mathcal D_g^{(t)}(u)=\mathcal D_g^{(t)}(0)-u$ for $u\le0$. If $A\ge\mathcal D_g^{(t)}(0)$, the generalized-inverse boundary is therefore $u=\mathcal D_g^{(t)}(0)-A\le0$.
\end{proof}
\begin{bluerevision}
Defining $\tilde\rho_{A,g}^{(t)}$ as the infimum of acceptable capital adjustments restores translation invariance despite the positive part in \eqref{eq:acceptance}. Because the Choquet integral generated by a concave $g$ is coherent, this acceptance-set definition yields the convex monetary risk measure proved below.
\end{bluerevision}

\begin{Proposition}\label{prop:convex_axioms}
Assume that the distortion function $g$ is concave. For any tolerance level $A > 0$, the static risk measure $\tilde{\rho}_{A,g}^{(t)}$ defined in \eqref{eq:mincap} is a convex risk measure on $\mathcal{H}^0$.
\end{Proposition}

\begin{proof}
Let $M_t^X = \sup_{0 \le s \le t} X_s$ and $M_t^Y = \sup_{0 \le s \le t} Y_s$ denote the respective running suprema for any processes $X, Y \in \mathcal{H}^0$.

First, consider finiteness on $\mathcal H^\infty$. Choose $C<\infty$ such that $|M_t^X|\le C$ almost surely. For $u\ge C$, $(M_t^X-u)_+=0$, so the acceptance set is nonempty. Conversely,
\[
\E_g[(M_t^X-u)_+]\ge \E_g[M_t^X-u]
=\E_g[M_t^X]-u\ge-C-u.
\]
Acceptance therefore implies $u\ge-(C+A)$, so the acceptable set has a finite lower bound and $\tilde\rho_{A,g}^{(t)}(X)\in\R$.

Because $g$ is concave, the Choquet integral $\E_g[\cdot]$ is a subadditive operator satisfying monotonicity and positive homogeneity. We verify the remaining axioms for convexity:

\textit{Translation Invariance.} By definition, the risk measure evaluates the infimum of the set of acceptable capital levels for the translated process $X + c$. Letting $\mathcal{V} = \{ v \in \R : \E_g[(M_t^X - v)_+] \le A \}$, it is immediate that an evaluated capital $u$ for the translated process satisfies $u-c \in \mathcal{V}$. Applying the translation property of the infimum operator, we conclude $\tilde{\rho}_{A,g}^{(t)}(X+c) = \inf(\mathcal{V}) + c = \tilde{\rho}_{A,g}^{(t)}(X) + c$.

\textit{Monotonicity.} Since $X_s \le Y_s$ pointwise, we have $M_t^X \le M_t^Y$. Thus, for any $u \in \R$, $\E_g[(M_t^X - u)_+] \le \E_g[(M_t^Y - u)_+]$, which yields $\tilde{\rho}_{A,g}^{(t)}(X) \le \tilde{\rho}_{A,g}^{(t)}(Y)$ by definition of the infimum.

\textit{Convexity.} Let $\tilde{\rho}_X = \tilde{\rho}_{A,g}^{(t)}(X)$ and $\tilde{\rho}_Y = \tilde{\rho}_{A,g}^{(t)}(Y)$. Consider arbitrary acceptable capital levels $u_X > \tilde{\rho}_X$ and $u_Y > \tilde{\rho}_Y$ such that $\E_g[(M_t^X - u_X)_+] \le A$ and $\E_g[(M_t^Y - u_Y)_+] \le A$. For a convex combination $u = \alpha u_X + (1-\alpha)u_Y$ with $\alpha \in [0,1]$, we obtain:
\begin{align*}
\E_g\Big[\Bigl(\sup_{0 \le s \le t} (\alpha X_s + (1 - \alpha)Y_s) - u\Bigr)_+\Big] 
&\le \E_g\Big[\big(\alpha (M_t^X - u_X) + (1 - \alpha)(M_t^Y - u_Y)\big)_+\Big] \\
&\le \alpha \E_g\big[(M_t^X - u_X)_+\big] + (1 - \alpha) \E_g\big[(M_t^Y - u_Y)_+\big] \le A.
\end{align*}
Thus, $u$ is an acceptable capital level for the combined process, meaning $\tilde{\rho}_{A,g}^{(t)}\Bigl(\alpha X+(1-\alpha) Y\Bigr) \le u$. Taking the limit as $u_X \downarrow \tilde{\rho}_X$ and $u_Y \downarrow \tilde{\rho}_Y$ establishes convexity.
\end{proof}

\begin{bluerevision}
For the aggregate net-loss process $L$, the acceptance condition bounds $\mathcal D_g^{(t)}(u)$. Hence the fixed-tolerance value is its generalized inverse at $A$:
\end{bluerevision}
\begin{equation}\label{eq:inverse_operator_distorted}
\tilde{\rho}_{A,g}^{(t)}(L) = \inf\left\{u \in \R : \mathcal{D}_g^{(t)}(u) \le A\right\} := \big(\mathcal{D}_g^{(t)}\big)^{-1}(A).
\end{equation}

For any concave distortion, $g(x)\ge x$ on $[0,1]$, so $\mathcal D_g^{(t)}(u)\ge\mathcal D^{(t)}(u)$. Monotonicity and continuity from Lemma~\ref{lem:deficit_properties}, rather than strict decrease, are sufficient for the generalized-inverse ordering
\[
\tilde\rho_{A,g}^{(t)}(L)\ge \inf\{u\in\R:\mathcal D^{(t)}(u)\le A\}.
\]

\begin{Example}\label{ex:convex_exp}
Building upon the infinite-horizon exponential framework from Example~\ref{ex:coherent_exp}, we can analytically compute the fixed-tolerance monetary risk-measure value. With $C=(1-\mu R)^k/(kR)$, inversion of the complete piecewise deficit \eqref{eq:exp_deficit_piecewise} gives
\begin{equation}\label{eq:exp_convex}
\tilde\rho_{A,g}^{(\infty)}(L)=
\begin{cases}
\displaystyle \frac{1}{kR}\log\!\left(\frac{C}{A}\right),&0<A\le C,\\[3mm]
C-A,&A>C.
\end{cases}
\end{equation}
To benchmark the structural behavior of this measure, isolating the undistorted Expected Maximum Deficit (EMD) framework by setting $g(x)=x$ (i.e., $k=1$) collapses the convex measure to:
\begin{equation}\label{eq:exp_convex_undist}
\tilde\rho_A^{(\infty)}(L)=
\begin{cases}
\displaystyle \frac{1}{R}\log\!\left(\frac{1-\mu R}{RA}\right),&0<A\le \frac{1-\mu R}{R},\\[3mm]
\displaystyle \frac{1-\mu R}{R}-A,&A>\frac{1-\mu R}{R}.
\end{cases}
\end{equation}
\begin{liquidityrevision}
For $A>C$, inversion occurs on the negative-capital branch and gives $C-A$. The expected area in red of \citet{LoiselTrufin2014} supplies a complementary duration-sensitive criterion. Its tolerance has units of money times time, whereas $A$ has units of money; consequently, numerical rankings require an explicit conversion convention and no unconditional ordering follows from the two formulas.
\end{liquidityrevision}

\end{Example}

\subsubsection{A coherent measure with a proportional risk tolerance}
\label{subsubsec:proportional_coherent}

\begin{bluerevision}
While the fixed-tolerance measure $\tilde{\rho}_{A,g}^{(t)}$ restricts the absolute distorted expected positive shortfall by the scalar $A$, a capital-dependent tolerance prevents negative capital assignments and restores positive homogeneity. The resulting proportional-tolerance measure $\rho_{\delta,g}^{(t)}$ is therefore coherent.

For a generic continuous-time loss process $X \in \mathcal{H}^0$ and a proportionality parameter $\delta > 0$, define $\rho_{\delta,g}^{(t)}(X)$ as the minimal capital $u \in \R$ for which the distorted expected positive shortfall does not exceed $\delta$ times the additionally injected capital $u-X_0$, where $X_0 \in \R$ is deterministic:
\end{bluerevision}
\begin{equation}\label{eq:coherent_implicit}
\rho_{\delta,g}^{(t)}(X) = \inf \left\{ u \in \R \;\middle|\; \E_g\Big[\Bigl(\sup_{0\le s\le t} X_s - u\Bigr)_+\Big] \le \delta (u - X_0) \right\}.
\end{equation}

By linking the risk tolerance linearly to the added capital itself, we recover positive homogeneity and, consequently, full coherence.

\begin{Proposition}\label{prop:implicit_coherent}
Assume that the distortion function $g$ is concave and $\delta > 0$. The static risk measure $\rho_{\delta,g}^{(t)}$ defined in \eqref{eq:coherent_implicit} is a coherent risk measure on $\mathcal{H}^0$.
\end{Proposition}
\begin{proof}
Let $M_t^X = \sup_{0 \le s \le t} X_s$ and $M_t^Y = \sup_{0 \le s \le t} Y_s$ denote the respective running suprema for any processes $X, Y \in \mathcal{H}^0$.

First, we verify the finiteness condition. For a bounded process $X \in \mathcal{H}^\infty$, both its initial state $X_0$ and its running supremum $M_t^X$ are finite almost surely. Let $C \in \R$ be such that $M_t^X \le C$ almost surely. Because $X_0 \le M_t^X$ by definition (since $0 \in [0,t]$), it naturally holds that $C \ge X_0$. For any capital level $u \ge C$, the positive shortfall $(M_t^X - u)_+$ vanishes, yielding $\E_g[(M_t^X - u)_+] = 0$. Since $u \ge C \ge X_0$, the proportional margin $\delta(u - X_0)$ remains non-negative, rendering $u$ an acceptable capital level. Conversely, for any $u < X_0$, the proportional margin $\delta(u - X_0)$ becomes strictly negative. Since the expected positive shortfall $\E_g[(M_t^X - u)_+]$ is inherently non-negative, the acceptance condition cannot be met for $u < X_0$. Hence, the infimum over acceptable capital is bounded from below by $X_0$, ensuring $\rho_{\delta,g}^{(t)}(X) \in \R$.

Because $g$ is concave, the Choquet integral $\E_g[\cdot]$ is a subadditive operator satisfying monotonicity and positive homogeneity. We verify the remaining axioms of coherence:

\textit{Translation Invariance.} By definition, the risk measure for the translated process $X+c$ is given by $\inf \{ u \in \R \mid \E_g[(M_t^X + c - u)_+] \le \delta (u - (X_0 + c)) \}$. Applying the change of variables $v = u - c$, this is exactly equivalent to $\inf \{ v + c \in \R \mid \E_g[(M_t^X - v)_+] \le \delta (v - X_0) \}$, which evaluates to $\rho_{\delta,g}^{(t)}(X) + c$.

\textit{Monotonicity.} If $X_s\le Y_s$ for all $s$, then $M_t^X\le M_t^Y$ and $X_0\le Y_0$. If $u$ is acceptable for $Y$, then
\[
\E_g[(M_t^X-u)_+]\le \E_g[(M_t^Y-u)_+]
\le\delta(u-Y_0)\le\delta(u-X_0).
\]
Thus $u$ is acceptable for $X$, and taking infima gives $\rho_{\delta,g}^{(t)}(X)\le\rho_{\delta,g}^{(t)}(Y)$.

\textit{Subadditivity.} Let $\rho_X = \rho_{\delta,g}^{(t)}(X)$ and $\rho_Y = \rho_{\delta,g}^{(t)}(Y)$. Consider arbitrary acceptable capital levels $u_X > \rho_X$ and $u_Y > \rho_Y$. For the pooled capital $u = u_X + u_Y$, the subadditivity of the supremum and the positive part yields $(M_t^{X+Y} - u)_+ \le (M_t^X - u_X + M_t^Y - u_Y)_+ \le (M_t^X - u_X)_+ + (M_t^Y - u_Y)_+$. Applying the subadditivity of $\E_g[\cdot]$, we obtain:
\begin{align*}
\E_g\big[(M_t^{X+Y} - u)_+\big] &\le \E_g\big[(M_t^X - u_X)_+\big] + \E_g\big[(M_t^Y - u_Y)_+\big] \\
&\le \delta(u_X - X_0) + \delta(u_Y - Y_0) = \delta(u - (X_0 + Y_0)).
\end{align*}
Thus, $u$ is an acceptable capital level for $X+Y$, meaning $\rho_{\delta,g}^{(t)}(X+Y) \le u$. Taking the limit as $u_X \downarrow \rho_X$ and $u_Y \downarrow \rho_Y$ establishes subadditivity.

\textit{Positive Homogeneity.} For $\lambda > 0$, bounding $\E_g[(\lambda M_t^X - u)_+]$ by $\delta(u - \lambda X_0)$ is equivalent to bounding $\lambda \E_g[(M_t^X - u/\lambda)_+]$ by $\delta \lambda (u/\lambda - X_0)$. Dividing by $\lambda$ recovers the exact acceptance condition for $X$ with respect to the variable $u/\lambda$, which proves $\rho_{\delta,g}^{(t)}(\lambda X) = \lambda \rho_{\delta,g}^{(t)}(X)$. For $\lambda=0$, the identity holds trivially since $\rho_{\delta,g}^{(t)}(0) = 0$.
\end{proof}

\begin{bluerevision}
For the aggregate net-loss process $L$, $L_0=0$ by \eqref{eq:netloss}. The acceptance condition therefore compares $\mathcal D_g^{(t)}(u)$ with the proportional margin $\delta u$, giving
\end{bluerevision}
\begin{equation}\label{eq:implicit_netloss}
\rho_{\delta,g}^{(t)}(L) = \inf\left\{u \in \R \;\middle|\; \mathcal{D}_g^{(t)}(u) \le \delta u \right\}.
\end{equation}

By Lemma~\ref{lem:deficit_properties}, the deficit is continuous and nonincreasing, while $u\mapsto\delta u$ is strictly increasing. Since the left-hand side is nonnegative, no $u<0$ can be acceptable. Provided $\mathcal D_g^{(t)}(0)>0$, the minimum is the unique positive root $u^*$ satisfying
\begin{equation}\label{eq:implicit_root}
\mathcal{D}_g^{(t)}(u^*) = \delta u^*.
\end{equation}
The intersection balances the residual expected maximum deficit against the prescribed proportional margin $\delta u^*$. The economic meaning of $\delta$ must be specified by the user of the model; mathematically, the proportional form is what yields positive homogeneity. Furthermore, as previously observed, consistency with respect to the increasing convex order naturally carries over.

\begin{Theorem}[Consolidated increasing-convex-order consistency]\label{thm:consolidated_icx}
Let $M_t^{(1)}\licx M_t^{(2)}$ and let $g$ be concave. Then, simultaneously,
\begin{align*}
\mathcal D_{g,1}^{(t)}(u)&\le \mathcal D_{g,2}^{(t)}(u), &&u\in\R,\\
\rho_g^{(t)}(L^{(1)})&\le\rho_g^{(t)}(L^{(2)}),\\
\tilde\rho_{A,g}^{(t)}(L^{(1)})&\le\tilde\rho_{A,g}^{(t)}(L^{(2)}), &&A>0,\\
\rho_{\delta,g}^{(t)}(L^{(1)})&\le\rho_{\delta,g}^{(t)}(L^{(2)}), &&\delta>0.
\end{align*}
Under the compound-Poisson assumptions of Proposition~\ref{prop:lefevreetal}, all four conclusions follow from $Y^{(1)}\licx Y^{(2)}$.
\end{Theorem}
\begin{proof}
\begin{bluerevision}
The deficit inequality is Proposition~\ref{prop:distorted_lefevreetal}. Evaluation at $u=0$ gives $\rho_g^{(t)}(L^{(1)})\le\rho_g^{(t)}(L^{(2)})$. For both $\tilde\rho_{A,g}^{(t)}$ and $\rho_{\delta,g}^{(t)}$, every capital level acceptable for process~2 is acceptable for process~1; inclusion of the corresponding acceptance sets gives the two generalized-inverse inequalities. The compound-Poisson conclusion follows from Proposition~\ref{prop:lefevreetal}.
\end{bluerevision}
\end{proof}
\begin{bluerevision}
We now give the condition under which $\rho_{\delta,g}^{(t)}(L)$ exceeds the zero-reserve running-maximum value $\rho_g^{(t)}(L)$.
\end{bluerevision}

\begin{bluerevision}
\begin{Proposition}[Critical threshold]
\label{prop:critical_threshold}
Assume $0<u_c=\rho_g^{(t)}(L)<\infty$ and $\mathcal D_g^{(t)}(u_c)>0$. Then there exists a critical proportional margin $\delta^*>0$ such that, for $\delta<\delta^*$, the proportional-tolerance value exceeds the zero-reserve running-maximum value:
\begin{equation}
\rho_{\delta,g}^{(t)}(L) > \rho_g^{(t)}(L).
\end{equation}
The threshold is the ratio of the residual distorted EMD at $u_c$ to $u_c$:
\begin{equation}\label{eq:delta_star}
\delta^* = \frac{\mathcal{D}_g^{(t)}(u_c)}{u_c}.
\end{equation}
\end{Proposition}
\end{bluerevision}

\begin{proof}
At $u_c$, the proportional acceptance inequality holds exactly when $\delta\ge\mathcal D_g^{(t)}(u_c)/u_c$. If $\delta<\delta^*$, $u_c$ is not acceptable. Continuity of the deficit and strict increase of $u\mapsto\delta u-\mathcal D_g^{(t)}(u)$ then force the unique root to satisfy $u^*>u_c$. The positivity assumptions exclude cases such as a deterministic maximum, for which the residual deficit at $u_c$ can vanish and $\delta^*=0$.
\end{proof}

\begin{Example}\label{ex:prop_exp}
Following Example~\ref{ex:coherent_exp}, we obtain a closed-form expression for this proportional coherent risk measure under an infinite-horizon exponential distribution with proportional hazard distortion $g(x)=x^k$. Evaluating the intersection yields:
\[
\frac{(1-\mu R)^k}{kR} e^{-kRu^*} = \delta u^*.
\]
Solving for $u^*$ involves the Lambert $W$ function. Multiplying both sides by $kR e^{kRu^*}$ gives $kR u^* e^{kRu^*} = \frac{(1-\mu R)^k}{\delta}$. Since this takes the form $x e^x = y$, the proportional capital is:
\begin{equation}\label{eq:exp_lambert}
\rho_{\delta,g}^{(\infty)}(L) = \frac{1}{k R} W_0\left( \frac{(1-\mu R)^k}{\delta} \right),
\end{equation}
where $W_0(\cdot)$ denotes the principal branch of the Lambert $W$ function. For the undistorted case ($k=1$), this simplifies to $\rho_\delta^{(\infty)}(L) = \frac{1}{R} W_0\left( \frac{1-\mu R}{\delta} \right)$. This expression demonstrates how the proportional measure balances exponential severity decay ($1/R$) with a logarithmic dampening dependent on the tolerance limit $\delta$.

{To determine the critical threshold $\delta^*$ from Proposition~\ref{prop:critical_threshold} for a fixed $\mu > 0$, we evaluate the infimum $\inf_{R > 0} \delta^*(R)$ across valid infinite-horizon exponential risk profiles $R \in (0, 1/\mu)$. From Example~\ref{ex:coherent_exp}, applying the proportional hazard (PH) distortion $g(x) = x^k$ yields the coherent capital $u_c = \frac{(1-\mu R)^k}{kR}$. The residual deficit evaluates to $\mathcal{D}_g^{(\infty)}(u_c) = u_c e^{-(1-\mu R)^k}$. The ratio provides the threshold:
\begin{equation}
\delta^*_{\mathrm{PH}}(R) = e^{-(1-\mu R)^k}.
\end{equation}
Taking the limit as $R \downarrow 0$ yields the infimum:
\begin{equation}
\inf_{R > 0} \delta^*_{\mathrm{PH}}(R) = e^{-1}.
\end{equation}}
\end{Example}

\section{Optimal reserve allocation}
\label{sec:allocation}

\begin{liquidityrevision}
Let $L^{(1)},\ldots,L^{(K)}$ be line-specific loss processes and let a fixed total reserve $u>0$ be split as $u_k\ge0$, $\sum_{k=1}^K u_k=u$. Write
\[
R_s^{(k)}=u_k-L_s^{(k)},\qquad
M_t^{(k)}=\sup_{0\le s\le t}L_s^{(k)},\qquad
\psi_k(v,t)=\PP(M_t^{(k)}>v),
\]
and $\mathcal D_{g_k,k}^{(t)}(u_k)=\int_{u_k}^\infty g_k(\psi_k(v,t))\dd v$. We study two optimization objectives: the sum of marginal deficits, $\Jmarg$, and the expected largest residual line deficit, $\Jworst$. These are functions of the allocation vector, not monetary risk measures on a process.

\begin{bluerevision}
For the line reserves $R_s^{(k)}$, multivariate ruin theory distinguishes
\begin{align*}
E_{\rm OR}(\boldsymbol u)&=\bigcup_{k=1}^K\{M_t^{(k)}>u_k\},
&E_{\rm AND}(\boldsymbol u)&=\bigcap_{k=1}^K\{M_t^{(k)}>u_k\},\\
E_{\rm sim}(\boldsymbol u)&=\{\exists s\le t:R_s^{(k)}<0\ \text{for every }k\},
&E_{\rm pool}(u)&=\left\{\exists s\le t:\sum_{k=1}^K R_s^{(k)}<0\right\}.
\end{align*}
The AND event permits different breach times, whereas simultaneous ruin requires a common time. OR-based allocation appears in \citet{Loisel2004}, while \citet{GongEtAl2012} consider both OR and AND criteria; \citet{DelsingEtAl2019} study adaptive allocations subject to subset-ruin constraints. For a fixed total reserve $u$, $E_{\rm pool}(u)$ is independent of the split. As a criterion for failure of the pooled portfolio, $E_{\rm OR}$ may be too conservative when an isolated breach can be absorbed elsewhere, while $E_{\rm AND}$ may be too permissive because pooled ruin need not require every unit to breach.

The objectives below replace a single binary threshold by residual-severity criteria.

Method~1 integrates the marginal exceedance probabilities separately. It does not assume independence and avoids estimation of the joint law, but it cannot respond to cross-unit co-occurrence or dependence when the marginal laws are held fixed. Method~2 is a severity-integrated OR criterion: it integrates distorted OR probabilities over all additional deficit levels. It uses the joint law and is therefore sensitive both to dependence and to its modelling. Neither method measures pooled or simultaneous ruin. Both take the total reserve as exogenous and therefore differ from Euler or Aumann--Shapley allocation of a single aggregate risk measure \citep{Tsanakas2009}; Method~1 is related in allocation structure, though not in temporal aggregation, to integrated-negative-surplus allocations \citep{Loisel2005,LoiselTrufin2014}. Section~\ref{sec:implementation} gives their operational scope.
\end{bluerevision}
\end{liquidityrevision}

\subsection{Method 1: Allocation based on marginal distorted deficits}

The first objective is
$$\Jmarg(u_1,\ldots,u_K) = \sum_{k=1}^K \mathcal{D}_{g_k, k}^{(t)}(u_k),$$
subject to $\sum_{k=1}^K u_k=u$. It depends only on the marginal laws.

\begin{Proposition}\label{prop:allocation1}
Assume that for $k=1,\ldots,K$, the survival function $\psi_k(\cdot,t): [0,\infty)\to [0,1]$ is continuously differentiable and strictly decreasing. Further assume the distortion functions $g_k: [0,1] \to [0,1]$ are continuously differentiable and strictly increasing. For total reserve $u>0$, the objective is strictly convex and its unique optimal allocation $(u_1^*,\ldots,u_K^*)$ has a non-empty active subset $J\subseteq \{1,\ldots,K\}$ such that
\begin{equation}\label{eq:kkt_distorted}
g_k\big(\psi_k(u_k^*,t)\big) = g_j\big(\psi_j(u_j^*,t)\big),\quad \forall\, k,j\in J,
\end{equation}
and for all inactive lines $k \notin J$ and active lines $j \in J$, the boundary condition $g_k\big(\psi_k(0,t)\big) \le g_j\big(\psi_j(u_j^*,t)\big)$ holds with $u_k^*=0$.
\end{Proposition}

\begin{proof}
\begin{liquidityrevision}
The line-\(k\) derivative is
\(-g_k(\psi_k(u_k,t))\), which is strictly increasing; hence the objective is strictly convex. Continuity on the compact feasible simplex gives existence, and strict convexity gives uniqueness. For the Lagrangian
\(\Jmarg-\lambda(\sum_k u_k-u)-\sum_k\mu_k u_k\), stationarity gives
\(g_k(\psi_k(u_k^*,t))=-\lambda-\mu_k\). Complementary slackness sets \(\mu_k=0\) on the active set, proving \eqref{eq:kkt_distorted}; at an inactive coordinate, \(u_k^*=0\) and \(\mu_k\ge0\), which gives the stated boundary inequality.
\end{liquidityrevision}
\end{proof}

Equivalently, the allocation is a water-filling rule. For a common marginal threshold $\eta$,
\begin{equation}\label{eq:water_filling}
u_k^*(\eta)=
\begin{cases}
(g_k\circ\psi_k)^{-1}(\eta),&g_k(\psi_k(0,t))>\eta,\\
0,&g_k(\psi_k(0,t))\le\eta,
\end{cases}
\end{equation}
where $\eta$ is chosen uniquely so that $\sum_{k=1}^K u_k^*(\eta)=u$. Thus Proposition~\ref{prop:allocation1} is computationally implementable by a one-dimensional search for $\eta$.

\begin{Remark}[The Invariance Property]\label{rem:invariance}
Proposition~\ref{prop:allocation1} reveals an important structural property. If a uniform, strictly increasing distortion function is enforced globally across the firm ($g_k = g$ for all $k$), applying the unique inverse $g^{-1}$ to both sides of equation \eqref{eq:kkt_distorted} yields $\psi_k(u_k^*,t) = \psi_j(u_j^*,t)$. 

\begin{bluerevision}
Hence the optimizer is invariant to the choice of a common strictly increasing distortion, although the minimized objective value may change. In this case the point allocation also equalizes the unit ruin probabilities at their allocated reserves among the active units. The integrated-deficit objective nevertheless records residual severity and supplies an objective profile for assessing near-optimal allocations; its role is not to claim a universally different split from probability balancing. Heterogeneous distortions generally change the allocation.
\end{bluerevision}
\end{Remark}

\subsection{Method 2: Worst-line maximum deficit}

\begin{liquidityrevision}
Let $\boldsymbol u=(u_1,\ldots,u_K)$ and $Z(\boldsymbol u)=\max_{1\le k\le K}(M_t^{(k)}-u_k)_+$. The second objective is the distorted expectation of this largest residual line deficit:
\begin{equation}\label{eq:worst_tail_representation}
\Jworst(u_1,\ldots,u_K)
=\E_g[Z(\boldsymbol u)]
=\int_0^\infty g\!\left(\PP\!\left(\bigcup_{k=1}^K
\{M_t^{(k)}>u_k+v\}\right)\right)\dd v,
\end{equation}
again subject to $u_k\ge0$ and $\sum_k u_k=u$. Unlike $\Jmarg$, this objective depends on the joint law of the running maxima.
For brevity, write $\tilde\psi_v(\boldsymbol u)=\PP(\bigcup_k\{M_t^{(k)}>u_k+v\})$.
\end{liquidityrevision}

\begin{bluerevision}
At $v=0$, the probability in \eqref{eq:worst_tail_representation} is the usual OR-ruin probability. Integration over $v$ distinguishes a shallow breach from a deep one, so $\Jworst$ is not merely an OR probability evaluated at one threshold.
\end{bluerevision}

\begin{Lemma}\label{lem:rho2_convex}
If $g$ is concave, then $\Jworst$ is convex in $(u_1,\ldots,u_K)$.
\end{Lemma}
\begin{proof}
The argument of the expected value can be rewritten as $\sup_{0 \le s \le t} \max_{1 \le k \le K} (L_s^{(k)} - u_k)_+$. For any fixed realization of the claim paths, the mapping $u_k \mapsto (L_s^{(k)} - u_k)_+$ is convex. Since the maximum and supremum of convex functions are convex, the mapping $(u_1, \dots, u_K) \mapsto \sup_{0 \le s \le t}\max_{1 \le k \le K} (L_s^{(k)} - u_k)_+$ is pointwise convex. Finally, because the Choquet integral $\E_g[\cdot]$ generated by a concave distortion $g$ is a subadditive and positively homogeneous operator, it preserves the convexity of the functional. Therefore, $\Jworst$ is a convex function. Moreover, for allocation vectors $\boldsymbol u$ and $\boldsymbol w$,
\[
|Z(\boldsymbol u)-Z(\boldsymbol w)|\le\|\boldsymbol u-\boldsymbol w\|_\infty.
\]
Monotonicity and cash additivity of $\E_g$ consequently give
$|\Jworst(\boldsymbol u)-\Jworst(\boldsymbol w)|\le\|\boldsymbol u-\boldsymbol w\|_\infty$.
Thus the objective is continuous as well as convex, which supplies the continuity needed for the existence assertion below.
\end{proof}

\begin{Proposition}\label{prop:allocation2_opt}
Assume $u>0$, that $\boldsymbol u\mapsto\tilde\psi_v(\boldsymbol u)$ and $g$ are differentiable on the relevant ranges, and that differentiation under the improper integral is justified by an integrable dominating function. Define the positive marginal reserve benefit
\begin{equation}\label{eq:marginal_benefit}
b_k(\boldsymbol u)=-\frac{\partial\Jworst}{\partial u_k}(\boldsymbol u)
=\int_0^\infty g'\!\left(\tilde\psi_v(\boldsymbol u)\right)
\left[-D_k\tilde\psi_v(\boldsymbol u)\right]\dd v.
\end{equation}
An optimal allocation exists. For its non-empty active subset $J=\{k:u_k^*>0\}$,
\begin{equation}\label{eq:alloc2_opt}
b_k(\boldsymbol u^*)=b_j(\boldsymbol u^*)
\end{equation}
for all $k,j\in J$. For inactive $k\notin J$ and active $j\in J$,
\begin{equation}\label{eq:inactive_benefit}
b_k(\boldsymbol u^*)\le b_j(\boldsymbol u^*).
\end{equation}
Thus active lines have equal marginal reserve benefit, while an inactive line offers no larger benefit at the boundary.
\end{Proposition}

\begin{proof}
\begin{liquidityrevision}
Lemma~\ref{lem:rho2_convex} and compactness of the feasible simplex give existence, while convexity makes the KKT conditions sufficient. The domination assumption permits differentiation under the integral, and \eqref{eq:marginal_benefit} gives \(\partial_k\Jworst=-b_k\). Stationarity for
\(\Jworst-\lambda(\sum_k u_k-u)-\sum_k\mu_k u_k\) therefore makes \(b_k=-\lambda\) on active coordinates. On an inactive coordinate, complementary slackness gives \(\mu_k\ge0\) and \(b_k=-\lambda-\mu_k\), which proves \eqref{eq:inactive_benefit}.
\end{liquidityrevision}
\end{proof}

The differentiable statement is a convenient special case. The principal examples $g(x)=\sqrt{x}$ and TVaR are not continuously differentiable on all of $[0,1]$; the same optimality conditions then hold with convex subgradients and one-sided marginal benefits.

\begin{liquidityrevision}
For \(K\) allocation units and a fixed reserve \(U\), let
\[
\begin{aligned}
\Delta_{K-1}
&=\left\{\boldsymbol q\in[0,1]^K:
\sum_{k=1}^Kq_k=1\right\},\\
\Phi_{1,U}(\boldsymbol q)
&=\Jmarg(Uq_1,\ldots,Uq_K),
&
\Phi_{2,U}(\boldsymbol q)
&=\Jworst(Uq_1,\ldots,Uq_K),
\end{aligned}
\]
and 

\[
\Phi_{m,U}^{\star}=\min_{\boldsymbol q\in\Delta_{K-1}}
\Phi_{m,U}(\boldsymbol q),
\qquad m=1,2.
\]
For \(\ep\ge0\) and \(0<\Phi_{m,U}^{\star}<\infty\), define the relative
\(\ep\)-optimal set by
\begin{equation}\label{eq:near_optimal_set}
\mathcal A_{m,\ep}^{(K)}(U)
=\left\{\boldsymbol q\in\Delta_{K-1}:
\Phi_{m,U}(\boldsymbol q)
\le(1+\ep)\Phi_{m,U}^{\star}\right\},
\qquad m=1,2.
\end{equation}
Under the convexity conditions above, this is a non-empty compact convex subset of the simplex. It is an objective-profile stability diagnostic conditional on the specified model, not a confidence region. Its coordinate projections record the range of each share attained somewhere in the near-optimal set; their Cartesian product generally contains allocations that are not near-optimal.

For \(K=2\), identify \(\boldsymbol q=(q,1-q)\) with \(q\in[0,1]\), write
\[
J_{1,U}(x)=\Jmarg(x,U-x),
\qquad
J_{2,U}(x)=\Jworst(x,U-x),
\qquad
J_{m,U}^{\star}=\min_{0\le x\le U}J_{m,U}(x),
\]
and abbreviate \(\mathcal A_{m,\ep}^{(2)}(U)\) to
\(\mathcal A_{m,\ep}(U)\). This set is a non-empty closed interval. If the
minimizer \(x_m^*\) is unique and interior and
\(J_{m,U}''(x_m^*)=\kappa_m>0\), a Taylor expansion gives
\[
\operatorname{diam}\mathcal A_{m,\ep}(U)
=\frac{2}{U}\sqrt{\frac{2\ep J_{m,U}^{\star}}{\kappa_m}}
+o(\sqrt{\ep}),\qquad \ep\downarrow0.
\]
The tolerance \(\ep\) is chosen externally; when \(J_{m,U}^{\star}=0\), an additive tolerance may be used instead.
\end{liquidityrevision}

\subsubsection*{Independent portfolios}

For dependent portfolios, the union probability requires the joint law of the line maxima. Mutual independence supplies a tractable benchmark but excludes common shocks.

\begin{Proposition}[Ruin probability of the worst-line reserve]\label{prop:indep_ruin}
Assume the $K$ business lines are mutually independent. The union probability of the minimum reserve process translated by $v$ is
\begin{equation}\label{eq:agg_min_indep}
\tilde\psi_v(\boldsymbol u) = 1 - \prod_{k=1}^K \Big(1 - \psi_k(u_k + v; t)\Big).
\end{equation}
\end{Proposition}

\begin{proof}
By \eqref{eq:worst_tail_representation}, the complementary event is $\{M_t^{(k)}\le u_k+v\text{ for every }k\}$. Independence factors its probability into the product of the marginal distribution functions, which gives \eqref{eq:agg_min_indep}.
\end{proof}

\begin{liquidityrevision}
\begin{Example}[A uniquely selected exponential allocation]\label{ex:sharp_exponential}
Let the two maxima be independent with
\(\psi_k(v)=a_k e^{-bv}\) for \(v\ge0\), where \(a_k\in(0,1]\), and take the identity distortion. For a fixed total reserve \(U\), allocation \(x\) to line 1 gives
\begin{equation}\label{eq:exponential_profiles}
J_{1,U}(x)=\frac{a_1}{b}e^{-bx}+\frac{a_2}{b}e^{-b(U-x)},
\qquad
J_{2,U}(x)=J_{1,U}(x)-\frac{a_1a_2}{2b}e^{-bU}.
\end{equation}
Assume \(\lvert\log(a_1/a_2)\rvert<bU\). Strict convexity then gives both objectives the same unique interior allocation,
\begin{equation}\label{eq:exponential_unique_allocation}
x^*=\frac{U}{2}+\frac{1}{2b}\log\!\left(\frac{a_1}{a_2}\right),
\qquad q^*=\frac{x^*}{U}.
\end{equation}
It need not be an equal split: \(q^*\ne1/2\) whenever \(a_1\ne a_2\). Moreover,
\begin{equation}\label{eq:exponential_near_optimal_set}
\begin{aligned}
\frac{J_{1,U}(x^*+h)}{J_{1,U}(x^*)}
&=\cosh(bh), && x^*+h\in[0,U],\\
\mathcal A_{1,\ep}(U)
&=\left[q^*-\frac{\operatorname{arcosh}(1+\ep)}{bU},
q^*+\frac{\operatorname{arcosh}(1+\ep)}{bU}\right]\cap[0,1].
\end{aligned}
\end{equation}
Because the second term in \(J_{2,U}\) is constant in \(x\),
\(J_{2,U}-J_{2,U}^{\star}=J_{1,U}-J_{1,U}^{\star}\); since
\(J_{2,U}^{\star}<J_{1,U}^{\star}\), one also has
\(\mathcal A_{2,\ep}(U)\subseteq\mathcal A_{1,\ep}(U)\).
Thus both objectives uniquely select a possibly unequal allocation. For fixed
\(a_1/a_2\) and \(\ep>0\), both near-optimal sets become arbitrarily narrow as
\(bU\to\infty\).
\end{Example}
\end{liquidityrevision}

\begin{Example}[A dependent benchmark]\label{ex:dependent_worst_line}
Let $M_k=mI_k$, where $I_k$ is Bernoulli with $\PP(I_k=1)=p$ and $\PP(I_1=I_2=1)=q$. Thus the marginals are fixed while $q$ controls dependence. For the identity distortion, $0\le u_1\le u_2\le m$, and $Z=\max\{(M_1-u_1)_+,(M_2-u_2)_+\}$,
\begin{equation}\label{eq:dependent_bernoulli}
\Jworst(u_1,u_2)=\E[Z]
=p(m-u_1)+(p-q)(m-u_2).
\end{equation}
The objective depends explicitly on $q$: under perfect comovement $q=p$, only the larger shortfall contributes, whereas under independence $q=p^2$ the second term is positive. By contrast, $\Jmarg=p(m-u_1)+p(m-u_2)$ is unchanged. This elementary common-event model demonstrates the dependence sensitivity of Method~2 without imposing independence.
\end{Example}

\begin{bluerevision}
The worst-line objective targets the largest residual line-level shortfall but can understate a pooled deficit formed by several moderate unit deficits. A separate pooled constraint may therefore be imposed when pooled adequacy matters. Such a constraint can check or determine the total reserve $u$; in the present additive model, a pooled-deficit term cannot select its split once $u$ is fixed because the term is then constant.
\end{bluerevision}

\begin{Remark}[Impact of the Distortion Function]
In the preceding independent example, the identity distortion gives closed forms. A concave distortion changes the weighting of the joint tail and generally requires numerical integration or a subgradient method as described after Proposition~\ref{prop:allocation2_opt}. The direction and size of the resulting reallocation depend on the full joint distribution and should not be asserted without calculation.
\end{Remark}

\begin{liquidityrevision}
\section{Actuarial interpretation and multicoverage allocation}
\label{sec:implementation}

\begin{bluerevision}
Sections~\ref{sec:deficit}--\ref{sec:allocation} define the reserve-indexed deficit function $\mathcal D_g^{(t)}$, the risk measures $\rho_g^{(t)}$, $\tilde\rho_{A,g}^{(t)}$ and $\rho_{\delta,g}^{(t)}$, and the allocation objectives used by Methods~1 and~2 without assigning them an accounting or regulatory meaning. Three specifications of the loss process are particularly natural. First, an underwriting process accumulates claims and expenses less earned premiums and recoveries; its maximum is the greatest cumulative underwriting net loss. Second, if \(N_s^{\rm econ}\) is economic net assets, the choice \(L_s=N_0^{\rm econ}-N_s^{\rm econ}\) makes \(M_t\) the greatest deterioration in that position. Third, a process of dated cash outflows less inflows makes \(M_t\) the largest cumulative cash-flow gap. The first interpretation is developed below because the available data support a simplified illustration of it. The other two require, respectively, joint asset--liability valuations or transaction-dated payments and receipts. In all three cases, EMD is an internal path-risk criterion; an additional management policy is needed before it can be interpreted as an accounting provision, regulatory capital requirement or liquidity holding.
\end{bluerevision}

\begin{bluerevision}
If an independently selected benchmark $u^\star>0$, with $\mathcal D_g^{(t)}(u^\star)>0$, is to be reproduced by both the fixed- and proportional-tolerance measures, set
\end{bluerevision}
\begin{equation}\label{eq:tolerance_calibration}
 A=\mathcal D_g^{(t)}(u^\star),
 \qquad
 \delta=\frac{\mathcal D_g^{(t)}(u^\star)}{u^\star}.
\end{equation}
\begin{bluerevision}
These identities calibrate $A$ and $\delta$ to an external decision at $u^\star$; they do not identify that decision or risk appetite from the loss data. The parameters should be fixed before out-of-sample assessment and varied in sensitivity analysis.
\end{bluerevision}

\subsection{Underwriting net loss and reserve-budget allocation}
\label{subsec:underwriting_application}

For line \(k\), consider the occurrence-based underwriting process
\begin{equation}\label{eq:underwriting_loss_path}
L_{k,s}^{\rm UW}
=C_{k,s}^{\rm occ}+E_{k,s}^{\rm UW}
-P_{k,s}^{\rm earned}-H_{k,s}^{\rm rec},
\qquad 0\le s\le t,
\end{equation}
where the terms are cumulative claim amounts attributed to occurrences, other underwriting costs, earned premiums and recognized recoveries on a consistent recognition basis. Its running maximum includes both adverse outcomes that persist to \(t\) and adverse intermediate outcomes that are later offset.

\begin{bluerevision}
In this application, $u_k$ is the share assigned to unit $k$ from a fixed internal underwriting-risk budget $U$. We retain the term \emph{reserve allocation} used in Section~\ref{sec:allocation}, but $U$ need not equal an accounting claims reserve or prescribed solvency capital. If resources are fully fungible and transfers are instantaneous, costless and unconstrained, the split cannot affect the pooled reserve process. It may still serve internal attribution or limit setting; a physical unit-protection interpretation instead requires ring-fencing, transfer delays, transfer costs or other managerial constraints. Method~1 allocates $U$ by balancing marginal reductions in the unit deficits. Method~2 targets the distorted expected largest residual unit deficit and therefore uses the joint predictive distribution. The present occurrence-based implementation requires unit identifiers, occurrence dates and severities, a premium-earning or other funding path and, for Method~2, a defensible dependence model. The public data used below supply dated, jointly observed loss components but not an observed premium or funding path.
\end{bluerevision}

\subsection{Monte Carlo estimation from predictive paths}
\label{subsec:monte_carlo}

Let \(L^{(1)},\ldots,L^{(B)}\) be paths simulated from a fitted predictive model conditional on the information available at the valuation date, and define
\begin{equation}\label{eq:simulated_maxima_estimator}
m_b=\sup_{0\le s\le t}L_s^{(b)},
\qquad
\widehat\psi_t(v)=\frac{1}{B}\sum_{b=1}^B\ind\{m_b>v\}.
\end{equation}
The historical observations fit or calibrate the model; they are not themselves the \(B\) predictive paths. Each simulation generates a possible future count, event times and severities, combines them with the modelled funding path, and yields one maximum. Thus a dataset is needed to construct the predictive law, while simulation supplies the repeated future paths that a short observed history cannot provide.

The plug-in estimator is
\begin{equation}\label{eq:empirical_emd_estimator}
\widehat{\mathcal D}_{g}^{(t)}(u)
=\int_u^\infty g\!\left(\widehat\psi_t(v)\right)\dd v.
\end{equation}
If \(m_{(1)}\le\cdots\le m_{(B)}\) are the ordered maxima and
\(r=\min\{j:m_{(j)}>u\}\), then, when this set is nonempty,
\begin{equation}\label{eq:empirical_emd_order_statistics}
\widehat{\mathcal D}_{g}^{(t)}(u)
=\sum_{j=r}^{B}
\bigl(m_{(j)}-m_{(j-1)}^{\circ}\bigr)
g\!\left(\frac{B-j+1}{B}\right),
\quad
m_{(r-1)}^{\circ}=u,
\quad m_{(j-1)}^{\circ}=m_{(j-1)}\quad(j>r),
\end{equation}
and the estimate is zero if \(u\ge m_{(B)}\). Event-driven simulations evaluate the supremum at every event time; a discretized model requires a sufficiently fine grid.
\begin{bluerevision}
The zero-reserve risk-measure estimate is $\widehat\rho_g^{(t)}=\widehat{\mathcal D}_{g}^{(t)}(0)$. Numerical inversion of $\widehat{\mathcal D}_{g}^{(t)}(u)=A$ and $\widehat{\mathcal D}_{g}^{(t)}(u)=\delta u$ gives the estimates of $\tilde\rho_{A,g}^{(t)}$ and $\rho_{\delta,g}^{(t)}$, respectively. For several allocation units, paths must be simulated jointly and retained scenario by scenario: Method~1 uses the marginal maxima, whereas Method~2 uses their joint vector. Repeated simulation batches assess Monte Carlo error; resampling or posterior predictive methods are needed for parameter uncertainty.
\end{bluerevision}

\subsection{A three-coverage allocation illustration with \texttt{danishmulti}}
\label{subsec:danishmulti}

The public \texttt{danishmulti} data contain \(2{,}167\) commercial-fire losses recorded by Copenhagen Reinsurance from 1980 to 1990 \citep{EmbrechtsEtAl1997,DutangCharpentier2026}. Amounts are expressed in millions of Danish kroner, adjusted to 1985 values, and each event is divided into Building, Contents and loss-of-Profits components. The observations have total loss of at least one million kroner. The three components are coverage-group allocation units affected by a common fire, not independent statutory business lines.

We use the \(1{,}714\) events in 1980--1988 for fitting and retain the \(453\) events in 1989--1990 for predictive diagnostics. The annual training counts have mean \(\widehat\lambda=190.4\) and variance \(930.3\). In a method-of-moments Poisson--gamma fit, \(N_b\mid\Lambda_b\) is Poisson with intensity \(\Lambda_b\), while \(\Lambda_b\) is gamma with shape \(49.0\) and scale \(\widehat\lambda/49.0\); its fitted coefficient of variation is \(0.143\). Given the count, complete records \((T_j,Y_{j,B},Y_{j,C},Y_{j,P})\) are sampled jointly from the training observations. This preserves the empirical occurrence-time pattern and the within-event dependence among the three loss components. It is a compound frequency--severity specification with vector-valued marks \citep{KlugmanEtAl2019,YanLuJeong2024}.

The data contain neither premiums nor exposures. We therefore use the deterministic expected-cost funding rate
\[
c_k=\widehat\lambda\,\overline Y_k,
\]
where \(\widehat\lambda\) is the mean annual count and \(\overline Y_k\) is the training-sample mean of component \(k\). This is an explicit modelling assumption for the recorded large-loss layer, not an observed premium. For predictive path \(b\),
\begin{equation}\label{eq:danish_loss_path}
L_{k,s}^{(b)}
=\sum_{j=1}^{N_b}Y_{j,k}^{(b)}\ind\{T_j^{(b)}\le s\}-c_ks,
\qquad 0\le s\le1.
\end{equation}
The component rates are \(344.0\), \(241.2\) and \(45.1\) million kroner per year. Their shares, \(54.6\%\), \(38.3\%\) and \(7.2\%\), provide an expected-cost-proportional comparator. Equation~\eqref{eq:danish_loss_path} is occurrence-based: expenses, recoveries, settlement delays and asset returns are outside the data and the fitted path.

\begin{bluerevision}
The fitted predictive paths estimate the reserve-indexed deficit function, but the data identify neither the institution's monetary tolerance $A$ nor its proportional tolerance $\delta$. We therefore use neither $\tilde\rho_{A,g}^{(1)}$ nor $\rho_{\delta,g}^{(1)}$ to select the illustrative budget. With the identity distortion and $B=50{,}000$ joint predictive paths, the zero-reserve pooled estimate $\widehat{\mathcal D}_{\rm pool}^{(1)}(0)$ is used solely as a common reference:
\[
U_{\rm ref}=\widehat{\mathcal D}_{\rm pool}^{(1)}(0)
=\frac1B\sum_{b=1}^{B}M_{{\rm pool}}^{(b)},
\qquad
M_{{\rm pool}}^{(b)}=\sup_{0\le s\le1}\sum_{k=1}^3L_{k,s}^{(b)}.
\]
The estimate is $U_{\rm ref}=81.12$ million kroner, with Monte Carlo standard error $0.47$. This baseline-model score provides a common, model-based monetary scale for comparing allocation profiles; it is neither observed capital nor an amount obtained from a regulatory formula or either tolerance-based measure. We keep it fixed in the sensitivity calculations to isolate changes in the split, so stability of allocation shares does not imply stability of the pooled capital level.
\end{bluerevision}
Table~\ref{tab:danish_allocation} gives each numerical minimizer at this budget and the coordinate projections of \(\mathcal A_{m,0.005}^{(3)}(U_{\rm ref})\). The \(0.5\%\) tolerance is chosen for the illustration. The projections are not confidence intervals, and their Cartesian product is not the near-optimal set.

\begin{table}[H]
\centering
\scriptsize
\setlength{\tabcolsep}{3pt}
\caption{Three-coverage \texttt{danishmulti} allocation at \(U_{\rm ref}\). Entries are percentages. Point allocations and projections are in Building/Contents/Profits order; the final pair is the objective increase at equal/expected-cost allocation.}
\label{tab:danish_allocation}
\begin{tabular*}{0.98\textwidth}{@{\extracolsep{\fill}}lccccc@{}}
\toprule
& \shortstack{Point allocation\\B/C/P} & \shortstack{Building\\projection} & \shortstack{Contents\\projection} & \shortstack{Profits\\projection} & \shortstack{Comparator increase\\equal / expected cost}\\
\midrule
Method 1 & 41.4/47.3/11.3 & [33.8,49.0] & [39.9,54.9] & [7.2,16.5] & 7.1/1.7\\
Method 2 & 44.7/47.5/7.8  & [39.6,50.2] & [42.5,52.6] & [4.0,10.9] & 14.7/1.9\\
\bottomrule
\end{tabular*}
\end{table}

\begin{figure}[H]
\centering
\includegraphics[width=0.96\textwidth]{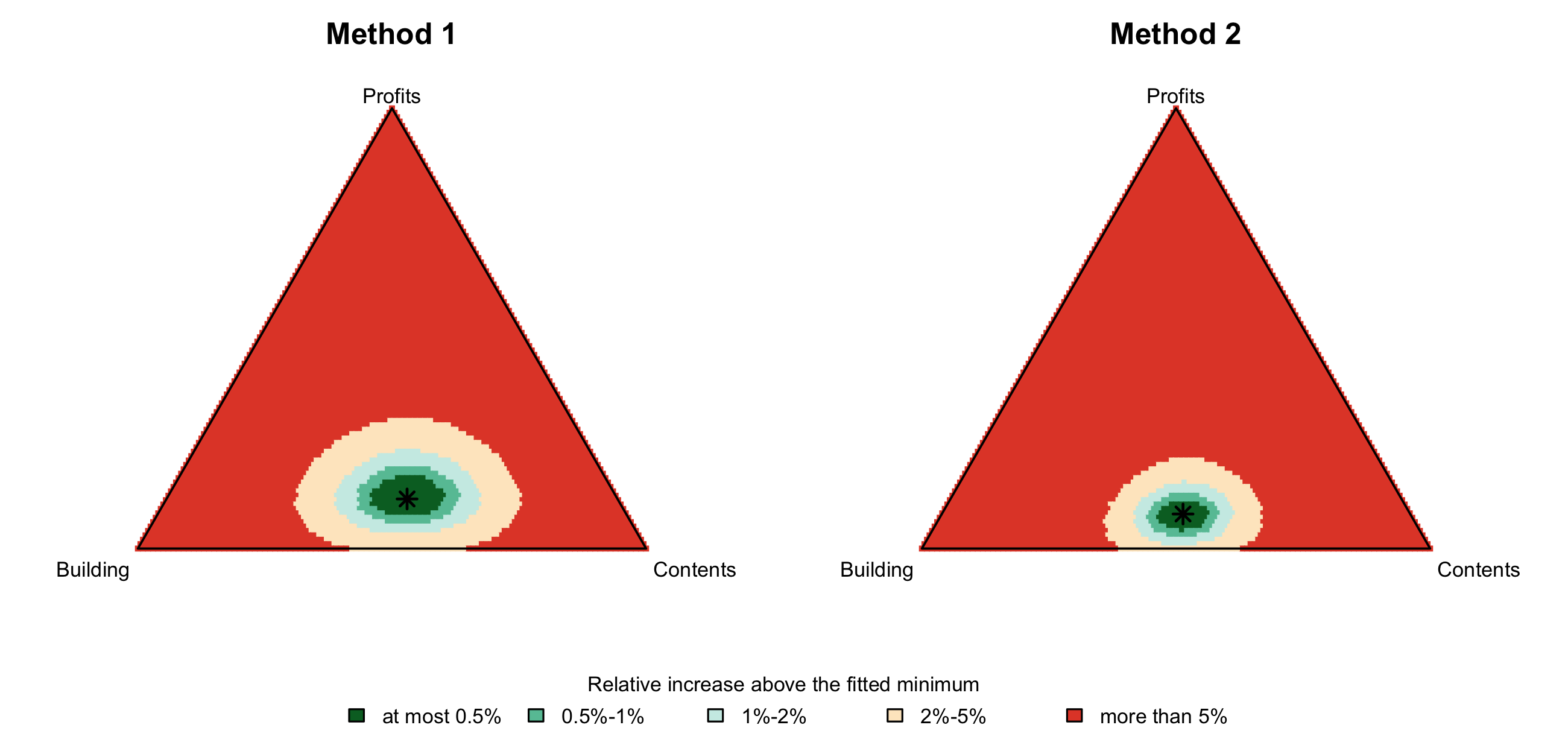}
\caption{Fitted allocation profiles on the three-coverage simplex. Colour records the relative increase above each fitted minimum, and the star marks the minimizing allocation. The dark-green region is the full 0.5\%-optimal set whose coordinate projections appear in Table~\ref{tab:danish_allocation}.}
\label{fig:danish_simplex}
\end{figure}

The figure and table show that neither fitted objective is flat over the simplex. Moving to equal allocation raises Method~1 by \(7.1\%\) and Method~2 by \(14.7\%\); the expected-cost-proportional comparator raises them by \(1.7\%\) and \(1.9\%\). The fitted model therefore distinguishes materially different allocations, while the near-optimal set displays the remaining local tolerance. Both point allocations are more balanced between Building and Contents than the expected-cost comparator. Relative to Method~1, Method~2 assigns a smaller Profits share, reflecting its fitted joint-maxima criterion rather than the sum of marginal deficits.

The two held-out years provide descriptive predictive diagnostics, not validation. Their observed counts, \(235\) and \(218\), lie inside the fitted 95\% predictive interval \([135,254]\). The observed pooled maxima are at the \(95.2\)nd and \(83.6\)th percentiles of the predictive distribution. In \(95.3\%\) of the fitted paths, the pooled running maximum exceeds the positive terminal loss, so the path criterion is not numerically reducible to a terminal-loss calculation in this example.

Sensitivity analysis separates identification under the fitted model from specification uncertainty. Replacing the mixed-Poisson count by a Poisson count, or empirical times by uniform times, changes an individual point share by at most \(3.6\) percentage points for Method~1 and \(5.6\) for Method~2. Omitting the single largest training event lowers the pooled EMD from \(81.1\) to \(69.0\) million kroner but, at the fixed budget \(U_{\rm ref}\), changes every allocation share by less than \(1.9\) percentage points. Increasing every funding rate by \(5\%\) moves the Method~2 allocation to \(40.7/49.9/9.4\). A marginal-preserving independence stress gives \(43.9/48.4/7.8\), whereas an intentionally severe reverse-rank coupling of Profits gives \(37.0/40.4/22.6\). The last coupling is a stress, not a fitted dependence model. These comparisons show that a locally informative objective profile does not eliminate tail, funding or dependence model risk \citep{Cairns2000,ASOP56}.

This remains an illustration rather than an operational validation. The data are historical, thresholded large losses; they contain occurrence rather than payment dates and only nine fitting years. Empirical resampling preserves observed multivariate loss records but cannot extrapolate beyond their observed tail; extreme-value or other tail models would require separate validation \citep{McNeil1997,YanLuJeong2024}. A real implementation would also require exposure or premium information and an application-specific funding path. The accompanying script \texttt{danishmulti\_emd\_allocation.R} reproduces the paths, allocation table, simplex figure and diagnostics.
\end{liquidityrevision}

\section{Conditional risk measures and review-date allocation}
\label{sec:dynamic}

\begin{bluerevision}
This section first defines the conditional running-maximum risk measure and proves a supermartingale inequality for it. It then defines the conditional fixed- and proportional-tolerance measures, shows by counterexample that the fixed-tolerance process need not be a supermartingale, establishes expected adequacy for an amount fixed at an earlier date, and recomputes the two reserve-allocation criteria at review dates.
\end{bluerevision}

\subsection{Conditional risk measures over a fixed horizon}
\label{subsec:timeconsistency}

For an integrable random variable $Z$, define the conditional distortion operator by
\begin{equation}\label{eq:conditional_distortion_operator}
\mathcal E_{g,s}(Z)
=\int_0^\infty g\!\left(\PP(Z>x\mid\mathcal F_s)\right)\dd x
+\int_{-\infty}^0\left[g\!\left(\PP(Z>x\mid\mathcal F_s)\right)-1\right]\dd x.
\end{equation}
The regular-conditional-distribution assumption made in Section~\ref{sec:deficit} ensures that the conditional survival functions can be chosen jointly measurable. All statements below assume the displayed conditional integrals are finite, so that the processes are integrable rather than merely extended-valued.

\vspace{3mm}
\noindent\textbf{The conditional running-maximum risk measure.}
\vspace{1mm}

\noindent For $0\le s\le t$, define $\rho_{g,s}^{(t)}(L)=\mathcal E_{g,s}(M_t)$. Since $M_t\ge0$, this becomes
\begin{equation}\label{eq:cond_risk}
\rho_{g,s}^{(t)}(L) = \int_0^\infty g\Big(\mathbb{P}\big(M_t > v \mid \mathcal{F}_s\big)\Big) \dd v.
\end{equation}

Nonlinear conditional distortion operators generally lack a recursive tower property. We therefore prove the supermartingale inequality in Proposition~\ref{prop:supermartingale}, without identifying it with strong time consistency.

\begin{Proposition}[Supermartingale inequality]
\label{prop:supermartingale}
Let $g$ be concave, not necessarily strictly concave, and assume $(\rho_{g,s}^{(t)}(L))_{0\le s\le t}$ is integrable. For $0\le s\le r\le t$,
\begin{equation}
\rho_{g,s}^{(t)}(L) \ge \mathbb{E}\big[\rho_{g,r}^{(t)}(L) \mid \mathcal{F}_s\big].
\end{equation}
\end{Proposition}

\begin{proof}
By definition, the conditional risk measure evaluated at a future time $r \in [s, t]$ is:
\begin{equation*}
\rho_{g,r}^{(t)}(L) = \int_0^\infty g\Big(\mathbb{P}\big(M_t > v \mid \mathcal{F}_r\big)\Big) \dd v.
\end{equation*}
Taking the conditional expectation with respect to the current filtration $\mathcal{F}_s$ and applying the Fubini-Tonelli theorem to interchange the integration and the expectation operators yields:
\begin{equation*}
\mathbb{E}\big[\rho_{g,r}^{(t)}(L) \mid \mathcal{F}_s\big] = \int_0^\infty \mathbb{E}\left[ g\Big(\mathbb{P}\big(M_t > v \mid \mathcal{F}_r\big)\Big) \;\Bigg|\; \mathcal{F}_s \right] \dd v.
\end{equation*}
Concavity of $g$ is sufficient. Applying conditional Jensen's inequality and the tower property ($\mathcal{F}_s\subseteq\mathcal{F}_r$) gives
\begin{align*}
\mathbb{E}\left[ g\Big(\mathbb{P}\big(M_t > v \mid \mathcal{F}_r\big)\Big) \;\Bigg|\; \mathcal{F}_s \right] &\le g\left( \mathbb{E}\Big[ \mathbb{P}\big(M_t > v \mid \mathcal{F}_r\big) \;\Big|\; \mathcal{F}_s \Big] \right) \\
&= g\Big(\mathbb{P}\big(M_t > v \mid \mathcal{F}_s\big)\Big).
\end{align*}
Integrating this bound over the threshold domain $[0, \infty)$ reconstructs the conditional risk measure at time $s$:
\begin{equation*}
\mathbb{E}\big[\rho_{g,r}^{(t)}(L) \mid \mathcal{F}_s\big] \le \int_0^\infty g\Big(\mathbb{P}\big(M_t > v \mid \mathcal{F}_s\big)\Big) \dd v = \rho_{g,s}^{(t)}(L).
\end{equation*}
This establishes the claimed inequality. Identity and TVaR distortions show why strict concavity must not be imposed.
\end{proof}

\vspace{3mm}
\noindent\textbf{The conditional fixed- and proportional-tolerance risk measures.}
\vspace{1mm}

For $U\in L^0(\mathcal F_s)$, define the conditional deficit at an $\mathcal F_s$-measurable amount by
\begin{equation}\label{eq:conditional_deficit}
\mathcal D_{g,s}^{(t)}(U)
=\int_{-\infty}^{\infty}\ind_{\{U\le v\}}
g\!\left(\PP(M_t>v\mid\mathcal F_s)\right)\dd v.
\end{equation}
For deterministic $U=u$, this reduces to the integral from $u$ to infinity. The conditional fixed- and proportional-tolerance risk measures are $\mathcal F_s$-measurable generalized inverses defined by essential infima:
\begin{align}
\tilde\rho_{A,g,s}^{(t)}(L)
&=\operatorname*{ess\,inf}\left\{U\in L^0(\mathcal F_s):
\mathcal D_{g,s}^{(t)}(U)\le A\right\},\label{eq:conditional_fixed_essinf}\\
\rho_{\delta,g,s}^{(t)}(L)
&=\operatorname*{ess\,inf}\left\{U\in L^0(\mathcal F_s):
\mathcal D_{g,s}^{(t)}(U)\le\delta U\right\}.\label{eq:conditional_prop_essinf}
\end{align}

\begin{bluerevision}
\begin{Proposition}[The conditional fixed-tolerance risk-measure process need not be a supermartingale]\label{prop:root_counterexample}
Even for $g(x)=x$, $(\tilde\rho_{A,g,s}^{(t)}(L))_{0\le s\le t}$ need not satisfy a supermartingale inequality.
\end{Proposition}
\end{bluerevision}
\begin{proof}
\begin{bluerevision}
Take $A=1.5$, let $s$ precede the revelation of two equally likely branches, and let $r>s$ be the information date at which the branch is known. In branch 1, $M$ equals $0$ and $1$ with probabilities $0.25$ and $0.75$; in branch 2, $M$ equals $0$ and $12$ with probabilities $0.75$ and $0.25$. The two possible values of $\tilde\rho_{A,g,r}^{(t)}(L)$, obtained from $\E[(M-u)_+]=1.5$, are
\[
u_1=-0.75,\qquad u_2=6,
\]
so $\E[\tilde\rho_{A,g,r}^{(t)}(L)\mid\mathcal F_s]=2.625$. Before the branch is revealed,
\[
\PP(M=0)=0.5,\qquad\PP(M=1)=0.375,\qquad\PP(M=12)=0.125,
\]
and the current value is $\tilde\rho_{A,g,s}^{(t)}(L)=u_0=0.75$. Hence
\[
\tilde\rho_{A,g,s}^{(t)}(L)=0.75
<2.625
=\E[\tilde\rho_{A,g,r}^{(t)}(L)\mid\mathcal F_s],
\]
which contradicts the proposed supermartingale inequality. The failure arises from taking the essential infimum in \eqref{eq:conditional_fixed_essinf}, not from the conditional deficit itself.
\end{bluerevision}
\end{proof}

\begin{Proposition}[Expected adequacy of an amount fixed at an earlier date]\label{prop:frozen_capital}
Let $0\le s\le r\le t$ and let $U_s\in L^0(\mathcal F_s)$ satisfy the required integrability conditions. Then
\begin{equation}\label{eq:frozen_capital}
\E\!\left[\mathcal D_{g,r}^{(t)}(U_s)\mid\mathcal F_s\right]
\le \mathcal D_{g,s}^{(t)}(U_s).
\end{equation}
Consequently, if $U_s$ is acceptable under \eqref{eq:conditional_fixed_essinf}, the expected future deficit at the amount fixed at time $s$ remains at most $A$; if it is acceptable under \eqref{eq:conditional_prop_essinf}, it remains at most $\delta U_s$.
\end{Proposition}
\begin{proof}
Write $p_q(v)=\PP(M_t>v\mid\mathcal F_q)$. Conditional Tonelli, conditional Jensen and the tower property give
\begin{align*}
\E[\mathcal D_{g,r}^{(t)}(U_s)\mid\mathcal F_s]
&=\int_{-\infty}^{\infty}\ind_{\{U_s\le v\}}
\E[g(p_r(v))\mid\mathcal F_s]\dd v\\
&\le\int_{-\infty}^{\infty}\ind_{\{U_s\le v\}}g(p_s(v))\dd v
=\mathcal D_{g,s}^{(t)}(U_s).
\end{align*}
The indicator can be taken outside the conditional expectation precisely because $U_s$ is $\mathcal F_s$-measurable.
\end{proof}
We call \eqref{eq:frozen_capital} expected adequacy of an amount fixed at an earlier date; it is not asserted to be a recursive form of time consistency.

\subsection{Implementation at successive review dates}
\begin{liquidityrevision}
Let $0=t_0<t_1<\cdots$ be review dates and $\Delta_i=t_{i+1}-t_i$. For allocation unit $k$, define
\[
M_{k,i}^{(\Delta_i)}
=\sup_{0\le v\le\Delta_i}
\bigl(L_{t_i+v}^{(k)}-L_{t_i}^{(k)}\bigr),
\qquad
\psi_{k,i}(z)
=\PP\!\left(M_{k,i}^{(\Delta_i)}>z\mid\mathcal F_{t_i}\right).
\]
\begin{bluerevision}
$M_{k,i}^{(\Delta_i)}$ and $\psi_{k,i}$ describe only additional deterioration after $t_i$. Any deficit already realized by $t_i$ must enter the current state or a separate management constraint; it is not included in $M_{k,i}^{(\Delta_i)}$.
\end{bluerevision}

\begin{bluerevision}
\begin{Remark}[Terminal position and replenishment at a review date]\label{rem:terminal_replenishment}
Consider a review period $[0,t]$ during which a reserve amount $u\ge0$ is assigned at the outset and no capital is added or released before $t$. Under the cumulative-net-loss convention, the amount remaining immediately before review is $u-L_t$. If management then selects an $\mathcal F_t$-measurable target $U_t\ge0$, the signed adjustment needed to reach it is
\[
I_t=U_t-(u-L_t)=U_t-u+L_t.
\]
The amount injected is $(I_t)_+$, whereas $(-I_t)_+$ is released. If the target is restored to its initial level, $U_t=u$, then $I_t=L_t$. Thus $L_t>0$ produces an injection under that policy, although it does not mean that the original reserve has been exhausted; the pre-review position is negative only when $L_t>u$. If the target policy is specified at time~0, a premium principle $\Pi$ may be applied to $(I_t)_+$ to assign the ex ante funding amount $\Pi((I_t)_+)$. This is a separate management policy, not a missing component of EMD: $L_t$ is one of the values entering $M_t=\sup_{0\le s\le t}L_s$, so the two amounts should not be added unless they are combined through an explicit joint policy.
\end{Remark}
\end{bluerevision}

Given an exogenous $U_i\in L^0_+(\mathcal F_{t_i})$, choose $\mathcal F_{t_i}$-measurable allocations $u_{k,i}\ge0$ with $\sum_k u_{k,i}=U_i$ by solving, pointwise in the current information state,
\begin{equation}\label{eq:review_date_allocation}
\min_{u_{k,i}\ge0,\,\sum_k u_{k,i}=U_i}
\sum_{k=1}^K\int_{u_{k,i}}^\infty g_k(\psi_{k,i}(z))\dd z.
\end{equation}
When the regularity conditions of Proposition~\ref{prop:allocation1} hold almost surely, its KKT characterization applies pathwise. Method~2 instead uses
\[
\tilde\psi_{v,i}(\boldsymbol u_i)
=\PP\!\left(\bigcup_{k=1}^K
\{M_{k,i}^{(\Delta_i)}>u_{k,i}+v\}
\,\middle|\,\mathcal F_{t_i}\right)
\]
inside the conditional counterpart of \eqref{eq:worst_tail_representation}. \begin{bluerevision}Equation~\eqref{eq:review_date_allocation} and the conditional Method~2 program describe periodic re-estimation and reallocation of an exogenous budget; they are not an intertemporal control model. Translating a refreshed EMD value into underwriting resources, economic net-asset headroom or usable liquidity requires an application-specific management policy.\end{bluerevision}
\end{liquidityrevision}


\section*{Acknowledgements}
P. Zuyderhoff acknowledges the School of Mathematical Sciences and the Faculty of Science and Engineering at UNNC for their institutional support. The work of C. Lefèvre was conducted within the Research Chair DIALog under the aegis of the Risk Foundation, a joint initiative by Université Paris-Dauphine PSL and CNP Assurances.


\end{document}